\documentclass[a4paper,onecolumn,11pt,accepted=2023-12-25]{quantumarticle}

\pdfoutput=1
\usepackage[affil-it]{authblk}
\usepackage[usenames,dvipsnames]{xcolor}
\usepackage{amsfonts}
\usepackage{amsmath,amsthm,amssymb,dsfont,pifont}
\usepackage{enumerate}
\usepackage[english]{babel}
\usepackage{graphicx}	
\usepackage{subcaption}
\usepackage[margin=3cm]{geometry}
\usepackage{url}
\usepackage{todonotes}
\usepackage{bbm}
\usepackage{caption}
\usepackage{subcaption}
\usepackage{tikz}
\usetikzlibrary{chains}
\usetikzlibrary{fit}
\usetikzlibrary{quantikz}
\usepackage{makecell}
\usepackage{cite}

\usepackage{epsfig}
\usetikzlibrary{shapes.symbols,patterns} 
\usepackage{pgfplots}

\usepackage{hyperref}[breaklinks]
\hypersetup{colorlinks=true,citecolor=blue,linkcolor=blue,filecolor=blue,urlcolor=blue}

\usepackage{nicefrac}
\usepackage{mathtools}

\usepackage{algorithm}
\usepackage{algorithmic}

\usepackage{optidef}

\usepackage{bold-extra}

\tikzset{meter/.append style={draw, inner sep=8, rectangle, font=\vphantom{A}, minimum width=30, line width=.8,
 path picture={\draw[black] ([shift={(.1,.3)}]path picture bounding box.south west) to[bend left=50] ([shift={(-.1,.3)}]path picture bounding box.south east);\draw[black,-latex] ([shift={(0,.1)}]path picture bounding box.south) -- ([shift={(.3,-.1)}]path picture bounding box.north);}}}
 
\theoremstyle{plain}
\newtheorem{theorem}{Theorem}
\newtheorem{lemma}[theorem]{Lemma}

\newtheorem{corollary}[theorem]{Corollary}

\theoremstyle{definition}

\newtheorem{remark}[theorem]{Remark}

\newtheorem{assumption}[theorem]{Assumption}

\newcommand*{\Id}{\mathrm{id}}

\newcommand*{\E}{\mathbb{E}}
\newcommand*{\Var}{\mathrm{Var}}

\newcommand*{\cE}{\mathcal{E}}

\newcommand*{\cL}{\mathcal{L}}

\newcommand*{\cO}{\mathcal{O}}

\newcommand*{\cR}{\mathcal{R}}
\newcommand*{\cS}{\mathcal{S}}

\newcommand*{\cU}{\mathcal{U}}

\newcommand*{\cW}{\mathcal{W}}

\newcommand*{\N}{\mathbb{N}}

\newcommand*{\R}{\mathbb{R}}

\newcommand*{\eps}{\varepsilon}
\newcommand*{\diag}{\mathrm{diag}}

\newcommand{\norm}[1]{\left\lVert#1\right\rVert}

\definecolor{mylightgreen}{RGB}{219,255,192}
\definecolor{mylightblue}{RGB}{240,255,252}
\definecolor{mylightyellow}{RGB}{255,248,180}
\definecolor{mylightorange}{RGB}{252,248,236}
\definecolor{mygraywhite}{RGB}{243, 243,243}
\definecolor{mylightred}{RGB}{255, 209,192}
\definecolor{mylightpink}{RGB}{255,240,252}

\usepackage{float}
\DeclareMathOperator{\sign}{sign}
\newcommand{\Pegasos}{\textsc{Pegasos}}
\usepackage{physics}
\usepackage{algorithmic}
\usepackage{algorithm}
\usepackage{cleveref}

\renewcommand{\arraystretch}{1.25}

\usepackage{tikz}
\usetikzlibrary{quantikz}
\usetikzlibrary{shapes.symbols,patterns} 
\usepackage{pgfplots}
\usetikzlibrary{shapes.symbols,patterns} 

\definecolor{myBlue}{RGB}{31,119,180}
\definecolor{myOrange}{RGB}{255, 127, 14}
\definecolor{myGreen}{RGB}{44,160,44}
\definecolor{myRed}{RGB}{214,39,40}
\definecolor{myViolet}{RGB}{148,103,189}
\definecolor{myGrey}{RGB}{127,127,127}
\definecolor{myCyan}{RGB}{23, 190, 207}

\definecolor{featureMap}{RGB}{23, 190, 207}
\definecolor{variationalForm}{RGB}{148,103,189}
\definecolor{highlight}{RGB}{214,39,40}

\allowdisplaybreaks
\crefname{equation}{}{}
\title{The complexity of quantum support vector machines}
\author{Gian Gentinetta}
\affiliation{Institute of Physics, \'Ecole Polytechnique F\'ed\'erale de Lausanne}
\affiliation{IBM Quantum, IBM Research Europe -- Zurich}
\orcid{0000-0002-5891-3289}
\author{Arne Thomsen}
\affiliation{Department of Physics, ETH Zurich}
\affiliation{IBM Quantum, IBM Research Europe -- Zurich}
\orcid{0000-0002-0309-9021}
\author{David Sutter}
\affiliation{IBM Quantum, IBM Research Europe -- Zurich}
\thanks{dsu@zurich.ibm.com; Gian Gentinetta and Arne Thomsen contributed equally and are listed in alphabetical order.}
\orcid{0000-0001-9779-8888}
\author{Stefan Woerner}
\affiliation{IBM Quantum, IBM Research Europe -- Zurich}

\begin{document}
\maketitle

\begin{abstract}
Quantum support vector machines employ quantum circuits to define the kernel function. It has been shown that this approach offers a provable exponential speedup compared to any known classical algorithm for certain data sets. The training of such models corresponds to solving a convex optimization problem either via its primal or dual formulation. Due to the probabilistic nature of quantum mechanics, the training algorithms are affected by statistical uncertainty, which has a major impact on their complexity. We show that the dual problem can be solved in $\cO(M^{4.67}/\eps^2)$ quantum circuit evaluations, where $M$ denotes the size of the data set and $\eps$ the solution accuracy compared to the ideal result from exact expectation values, which is only obtainable in theory. We prove under an empirically motivated assumption that the kernelized primal problem can alternatively be solved in $\cO(\min \{ M^2/\varepsilon^6, \, 1/\eps^{10} \})$ evaluations by employing a generalization of a known classical algorithm called \Pegasos. Accompanying empirical results demonstrate these analytical complexities to be essentially tight. In addition, we investigate a variational approximation to quantum support vector machines and show that their heuristic training achieves considerably better scaling in our experiments.
\end{abstract}

\maketitle

\section{Introduction}
Finding practically relevant problems where quantum computation offers a speedup compared to the best known classical algorithms is one of the central challenges in the field. Quantifying a speedup requires a provable convergence rate of the quantum algorithms, which restricts us to studying algorithms that can be analyzed rigorously.  
The impressive recent progress on building quantum computers gives us a new possibility: We can use heuristic quantum algorithms that can be run on current devices to demonstrate the speedup empirically. This however requires a hardware friendly implementation, i.e., a moderate number of qubits and shallow circuits.

In recent years, more and more evidence has been found supporting machine learning tasks as good candidates for demonstrating quantum advantage~\cite{QML_lloyd17,Havlicek2019,Abbas2020a,Liu2021}. 
In particular, the so-called \emph{supervised learning} setting, where in the simplest case the goal is to learn a binary classifier of classical data, received much attention. The reasons are manifold:
\begin{enumerate}[(i)]
    \item The algorithms only require \emph{classical access to data}. This avoids the common assumption that data is provided in the amplitudes of a quantum state which is hard to justify~\cite{Aaronson2015}. It has been observed that if classical algorithms are provided with analogous sampling access to data, claimed quantum speedups may disappear~\cite{tan19,tan20}.
    \item Kernelized support vector machines offer a framework that can be analyzed rigorously. In addition, this approach can be immediately lifted to the quantum setting by using a quantum circuit parametrized by the classical data as the feature map. This then defines a \emph{quantum support vector machine} (QSVM)~\cite{Havlicek2019, Liu2021}.
    \item It is known that, for certain (artificially constructed) data sets, QSVMs can offer a provable exponential speedup compared to any known classical algorithm~\cite{Liu2021}. Note, however, that according to current knowledge such QSVMs require deep circuits and, hence, rely on a fault-tolerant quantum computer.
\end{enumerate}
Despite the growing interest in QSVMs, a rigorous complexity-theoretic treatment comparing the dual and primal approaches to training models employing arbitrary quantum kernels is missing (the analysis presented in \cite{li2019} considers kernel functions defined by quantum states, but these are limited to classical polynomial- and RBF-kernels).
In contrast to the classical case where the kernel function can be computed exactly, the runtime analysis of QSVMs is complicated by the fact that all kernel evaluations are unavoidably subject to statistical uncertainty stemming from finite sampling, even on a fully error corrected quantum computer. This is the case since the quantum kernel function is defined as an expectation value, which in practice can only be approximated as a sample mean over a finite number of measurement shots. 
Furthermore, it has recently been shown that quantum kernels can suffer from exponential concentration, if one is not careful with the construction of the feature map \cite{Thanasilp2022}. Then, an exponential number of measurements is necessary to distinguish the kernel evaluations. However, in this work, we focus on the training of the QSVM and, hence, assume that the feature map which embeds the data through a quantum circuit has been chosen reasonably (see assumption on noisy halfspace learning in \cite[Lemma 14, Definition 15]{Liu2021}).
Being able to successfully train these models is a necessary prerequisite for the ultimate goals of quantum advantage and good generalization, which is why we think the training stage of the algorithms warrants a detailed investigation of its own.

The computational complexity of QSVMs is then defined as the total number of quantum circuit evaluations\footnote{Given a fixed feature map, the number of circuit evaluations is computationally dominant and therefore this definition of complexity is equivalent to the total runtime of finding such a solution.} necessary to find a decision function which, when evaluated on the training set, is $\eps$-close to the solution employing an infinite number of measurement shots. Note that this definition solely in terms of training is legitimate since fitting the models is computationally dominant compared to evaluating the final classifier.


The beauty of QSVMs is that they allow the classification task to be characterized by an efficiently solvable convex optimization problem. Classically, the problem is especially simple in its dual form, where the kernel trick can be used to reduce the optimization to solving a quadratic program. The main drawback of the dual method is that the entire $M\times M$ kernel matrix needs to be calculated. When the kernel entries are affected by finite sampling noise, we find that the scaling is even steeper. In order to find an $\eps$-accurate classifier $\cO(M^{4.67}/\eps^2)$ quantum circuit evaluations are found to be necessary in the setting of \emph{noisy halfspace learning} \cite[Lemma 14, Definition 15]{Liu2021}.\footnote{This is an polynomial improvement compared to the finding in~\cite[Lemma~19]{Liu2021} which requires $\cO(M^{6}/\eps^2)$ circuit evaluations.} Here and throughout the manuscript, $\epsilon$ denotes an upper bound on the absolute difference between two fully trained classifiers, where one is idealized and built using exact expectation values, and the other is constructed with a finite number of quantum circuit evaluations, which result in sample means that are subject to noise. The dependence on $M$ poses a major challenge for problems with large data sets.

There are, however, two attempts to circumvent this potentially prohibitive scaling. Instead of solving the dual, we can make use of the probabilistic \Pegasos~algorithm~\cite{Shalev-Shwartz2008}, which minimizes the primal formulation of the problem using stochastic gradient descent. Unlike standard solvers of the primal, the algorithm can be kernelized. Its main advantage is that instead of requiring the whole kernel matrix, \Pegasos~is an iterative algorithm that only evaluates the kernel entries needed to compute the gradient in every step. We show that this results in $\cO(\min \{ M^2/\varepsilon^6, \, 1/\varepsilon^{10} \})$ measurement shots required to achieve an $\eps$-accurate classification. Note that this runtime can be independent of $M$. However, compared to the dual, this comes at the cost of worse scaling in $\eps$.

A second approach is to employ a variational quantum circuit containing $d$ parameters that are trained to minimize a loss function evaluated on the training set. It has been shown that just like QSVMs, such models implement a linear decision boundary in feature space and can therefore be viewed as approximate QSVMs~\cite{Havlicek2019, Arne}. For an appropriately chosen ansatz and a suitable training algorithm, the approximation quality gets better with increasing $d$. For this heuristic approach to training, we find an empirical scaling of $\cO(1/\varepsilon^{2.9})$, which is again independent of $M$ due to the use of stochastic gradient descent and marks a significant improvement over \Pegasos\ with respect to $\eps$. The price to pay for this reduction in complexity is that the underlying optimization problem is non-convex and its solution therefore loses the guarantee of global optimality.

\paragraph{Results:} Our findings are summarized in~\Cref{tab:qnn_qsvm_scaling}, where additional empirical scalings are included which show that the analytically bounds are essentially tight for practical data sets.
\begin{table}[!htb]
\centering

\bgroup
\def\arraystretch{1.5}

    \begin{tabular}{c | c | c | c}
			           & dual                           & primal (\Pegasos)                                                   & approx.~QSVM                            \\ \hline
              classical equivalent & $\cO(M^2)$ & $\cO(\min\{M/\varepsilon^2, \, 1/\varepsilon^4 \} )$& \ding{55}
              \\
			 analytical complexity & $\cO(M^{4.67}/\varepsilon^2)$ & $\cO(\min \{ M^2/\varepsilon^6, \, 1/\varepsilon^{10} \})$ & 	\ding{55}		    \\ 
			 empirical complexity &	 $\cO(M^{4.8 \pm 0.4}/\varepsilon^{2})$ & $\cO(1/\varepsilon^{9.5 \pm 1.0})$  &	$\cO(1/\varepsilon^{2.9 \pm 0.3})$ 		
			 
    \end{tabular} 
\egroup

		\caption{\textbf{Complexity of QSVM training,} i.e.~the asymptotic number of quantum circuit evaluations required to achieve an $\eps$-accurate decision function on a training set of size $M$.
}
		\label{tab:qnn_qsvm_scaling}
	\end{table}

In~\Cref{sec:svm} we give a brief overview of QSVMs, the optimization problems underlying their training, and the heuristic model designated approximate QSVM. The derivation of the analytical complexity statements for the dual and primal methods can be found in~\Cref{sec:complexity_dual,sec:complexity_pegasos}, respectively. In \Cref{sec:empirical_scaling}, we present the numerical experiments justifying the empirical complexity statements.



\section{Support vector machines}\label{sec:svm}
The machine learning task considered in this paper is the supervised training of binary classifiers on classical data using support vector machines (SVMs)~\cite{Vapnik1992,Cortes95,Vapnik2000}. Given an unknown probability distribution $P(\mathbf{x},y)$ for data vectors $\mathbf{x} \in \R^r$ and class membership labels $y \in \{-1,1\}$, we draw a set of training data $X = \{\mathbf{x}_1,\dots,\mathbf{x}_M\}$ with corresponding labels $y = \{y_1,\dots,y_M\}$. Using this set, the SVM defines a classification function $c: \R^r \to \{-1,1\}$ implementing the trade-off between accurately predicting the true labels and maximizing the orthogonal distance (called \emph{margin}) between the two classes. Once trained, the classification function can be applied to classify previously unseen data drawn from the same probability distribution. While this \emph{generalization performance} of models is crucial when solving a concrete machine learning task, the sole focus of this work lies elsewhere, namely in the training of the models. Therefore, throughout this work only a training and not a test is considered.

We start by introducing classical kernelized support vector machines. The quantum case then constitutes a straightforward extension where the feature maps are defined by quantum circuits. To conclude the section, we present variational quantum circuits and interpret them as approximate quantum support vector machines.

\subsection{Kernelized support vector machines}
Support vector machines can implement nonlinear decisions in the original data space by transforming the input vectors with a nonlinear feature map. Here, we consider the finite dimensional case where the feature map is given by $\varphi: \R^r \to \R^s$. The decision boundary implemented by some $\mathbf{w} \in \R^s$ is then given by
\begin{equation*}
    c_{\textnormal{SVM}}(\mathbf{x}) = \sign\big[\mathbf{w}^\top\varphi(\mathbf{x})\big],
\end{equation*}
which is linear in feature space.
We define the hyperplane $\mathbf{w}^\star$ as the solution to the primal optimization problem posed by support vector machines, which is given by
\begin{align}\label{opt:svm_soft_margin_primal_hinge_loss}
 \textnormal{(primal problem)} \qquad			\min_{\mathbf{w} \in \R^s} \left \lbrace \frac{\lambda}{2} \norm{\mathbf{w}}^2 +  \sum_{i=1}^M \max \left\lbrace 0, 1 - y_i \big(\mathbf{w}^\top \varphi(\mathbf{x}_i)\big) \right\rbrace \right \rbrace \, ,
\end{align}
where the term containing $\lambda>0$ provides regularization and the second term is a sum over the hinge losses of the data points contained in the training set.

Instead of solving the primal optimization problem, modern implementations~\cite{sklearn} usually optimize the dual\footnote{For technical reasons in the complexity analysis, we are not considering the exact dual of~\eqref{opt:svm_soft_margin_primal_hinge_loss} here. Instead of using the $\ell^1$ norm on the slack variables (one can show that this is equivalent to using the hinge loss like in~\Cref{opt:svm_soft_margin_primal_hinge_loss}) , the dual in \Cref{opt:svm_soft_margin_feature_map_dual} corresponds to the use of the $\ell^2$ norm in the primal. See \cite[Equation C6]{Liu2021} for a derivation.} of~\eqref{opt:svm_soft_margin_primal_hinge_loss}
\begin{align} \label{opt:svm_soft_margin_feature_map_dual}
\textnormal{(dual problem)} \qquad  \left \lbrace
\begin{array}{r l}
\max \limits_{\alpha_i \in \R} & \sum_{i=1}^M \alpha_i - \frac{1}{2} \sum_{i,j=1}^M \alpha_i \alpha_j y_i y_j \, k(\mathbf{x}_i, \mathbf{x}_j) - \frac{\lambda}{2} \sum_{i=1}^M \alpha_i^2 \\
\textnormal{s.t.} & 0 \leq \alpha_i \quad \forall i=1,\ldots,M
\end{array} \, .
\right.
\end{align}
The dual problem is typically favored because the \emph{kernel trick} can be straightforwardly applied by defining the symmetric, positive semidefinite kernel function 
\begin{equation}\label{eq:class_kernel}
    k(\mathbf{x},\mathbf{y}) := \varphi(\mathbf{x})^\top\varphi(\mathbf{y})\, ,\quad \mathbf{x}, \mathbf{y} \in \R^r
\end{equation}
as the inner product between feature vectors. 

The solutions of the primal and dual optimization problems are connected via the Karush-Kuhn-Tucker condition as $\mathbf{w} = \sum_{i=1}^M\alpha_iy_i\varphi(\mathbf{x}_i)$~\cite{Liu2021} and hence the classification function can be rewritten as
\begin{equation}
    c_{\textnormal{SVM}}(\mathbf{x}) 
    = \sign\left[\sum_{i=1}^M\alpha_i y_i \varphi(\mathbf{x}_i)^\top\varphi(\mathbf{x})\right] 
    = \sign\left[\sum_{i=1}^M\alpha_i y_i k(\mathbf{x}_i,\mathbf{x})\right] \, ,
    \label{eq:class_qsvm_dual}
\end{equation}
from which it is apparent that both the dual optimization problem and the classification function only depend on the kernel function and feature vectors need not be explicitly computed at any point. 


From~\Cref{opt:svm_soft_margin_feature_map_dual}, it follows that solving the dual requires the evaluation of the full kernel matrix $K \in \R^{M \times M}$ with entries 
\begin{equation}\label{eq:svm_kernel_entries}
	K_{ij} = k(\mathbf{x}_i, \mathbf{x}_j) \quad \textnormal{for} \quad i,j=1,\ldots,M \, .
\end{equation}
Note that given $K$, the dual optimization can be restated as a convex quadratic program (see~\Cref{opt:svm_l2_matrix_dual} in~\Cref{app:daniel}) and hence solved in polynomial time~\cite{Boyd2004}.

\subsection{Primal Estimated sub-GrAdient SOlver (\Pegasos)}
Alternatively, we can solve the primal problem with an algorithm called \Pegasos~\cite{Shalev-Shwartz2011a}\footnote{An implementation recently added to Qiskit~\cite{qiskit} can be found \href{https://qiskit.org/documentation/machine-learning/locale/ta_IN/stubs/qiskit_machine_learning.algorithms.PegasosQSVC.html}{here}.}, which arises from the application of stochastic sub-gradient descent on the objective function $f(\mathbf{w})$ in the primal optimization problem~\Cref{opt:svm_soft_margin_primal_hinge_loss}. Starting with initial weights $\mathbf{w}^0 = \mathbf{0}$, we iteratively optimize $\mathbf{w}^t$ for $t \in \{1,\dots,T\}$. In the following, we analyse how the weights are updated in every step.

As is standard in stochastic gradient descent, we start the optimization step $t$ by uniformly sampling a random index $i_t \in \{1,\dots,M\}$. We then define the partial objective $f^t$ as the hinge-loss for the chosen datum $\mathbf{x}_{i_t} \in X \subset \R^r$ added to the regularization term
\begin{equation*}
    f^t(\mathbf{w}) = \frac{\lambda}{2}\norm{\mathbf{w}}^2 + \max\left[0,1-y_{i_t}\mathbf{w}^T\varphi(\mathbf{x}_{i_t})\right],
\end{equation*}
where $\varphi: \R^r \to \R^s$ denotes the feature map and $\lambda$ is the regularization parameter. In order to find the steepest descent in the loss landscape defined by the partial objective $f^t$, we calculate the gradient of $f^t$ with respect to $\mathbf{w}$ as
\begin{equation*}
    \frac{\partial f^t}{\partial \mathbf{w}} = 
    \begin{cases}
        \lambda \mathbf{w}, & \textnormal{if } y_{i_t}\mathbf{w}^T\varphi(\mathbf{x}_{i_t}) > 1 \\
        \lambda \mathbf{w} - y_{i_t}\varphi(\mathbf{x}_{i_t}), & \textnormal{otherwise}.
    \end{cases}
\end{equation*}

Next, we update the weights
$
     \mathbf{w}^t = \mathbf{w}^{t-1} - \eta^t \left.\frac{\partial f^t}{\partial \mathbf{w}}\right|_{\mathbf{w}=\mathbf{w}^{t-1}},
$
where $\eta^t > 0$ is the learning rate. To further facilitate the calculation, we expand $\mathbf{w}^{t-1} = \frac{1}{\lambda(t-1)}\sum_{j=1}^M\alpha^{t-1}_jy_j\varphi(\mathbf{x}_j)$, leading to 
\begin{equation*}
    \mathbf{w}^t
    = \begin{cases}
        
        \frac{1 - \lambda\eta^t}{\lambda(t-1)}\displaystyle\sum_{j=1}^M\alpha^{t-1}_jy_j\varphi(\mathbf{x}_j), 
        & \textnormal{if } \frac{y_{i_t}}{t\lambda}\sum_{j=1}^M\alpha^{t-1}_jy_jk(\mathbf{x}_{i_t},\mathbf{x}_j) > 1 \\
        
         \frac{1 - \lambda\eta^t}{\lambda(t-1)}\displaystyle\sum_{j=1}^M\alpha^{t-1}_jy_j\varphi(\mathbf{x}_j) + \eta^t y_{i_t}\varphi(\mathbf{x}_{i_t}), 
         & \textnormal{otherwise}.
    \end{cases}
\end{equation*}
We can now fix the learning rate $\eta^t = 1/\lambda t$ to simplify the expression as 
\begin{equation*}
    \mathbf{w}^{t} = \frac{1}{t\lambda}\sum_{j=1}^M\alpha^{t}_jy_j\varphi(\mathbf{x}_j),
\end{equation*}
where the coefficients $\boldsymbol{\alpha}^t$ are defined as $\alpha^t_j = \alpha^{t-1}_j$ for $j \neq {i_t}$ and
\begin{equation*}
    \alpha^t_{i_t} = 
    \begin{cases}
        \alpha^{t-1}_{i_t}, & \textnormal{if } \frac{y_{i_t}}{t\lambda}\sum_{j=1}^M\alpha^{t-1}_jy_jk(\mathbf{x}_{i_t},\mathbf{x}_j) > 1 \\
        \alpha^{t-1}_{i_t} + 1, & \textnormal{otherwise}.
    \end{cases}
\end{equation*}
Crucially, the feature map is only accessed via the kernel entries, i.e. the \Pegasos\ algorithm allows for a kernelization of the primal optimization problem. While in theory we calculate the gradient of $f^t$ with respect to $\mathbf{w}$, in practice $\mathbf{w}$ is never explicitly calculated. Instead, it suffices to calculate and store the integer coefficients $\boldsymbol{\alpha}^t$ both for training and prediction.

\begin{algorithm}[!htb]
	\caption{Kernelized \Pegasos\ \cite[Figure 3]{Shalev-Shwartz2011a}} 
	\label{algo:kernelized_pegasos} 
	\begin{algorithmic}[1]
		\STATE \textbf{Inputs:}
		\STATE training data $T = \{\mathbf{x}_1, \mathbf{x}_2, ..., \mathbf{x}_M\}$
		\STATE labels $L = \{y_1, y_2, ..., y_M\}$
		\STATE regularization parameter $\lambda \in \R^+$
		\STATE number of steps $T \in \N$
		\STATE
		\STATE \textbf{Initialize:} $\boldsymbol{\alpha}_1 \gets \mathbf{0} \in \N^M$
		\STATE
		
		\FOR{$t = 1, 2, ..., T$}
		\STATE Choose $i_t \in \{0, ..., M\}$ uniformly at random.
		
		\FORALL{$j \neq i_t$}
		\STATE $\boldsymbol{\alpha}_{t+1}[j] \gets \boldsymbol{\alpha}_t[j]$
		\ENDFOR
		
		\IF{$y_{i_{t}} \frac{1}{\lambda t} \sum_{j=1}^M \boldsymbol{\alpha}_{t}[j] y_j \, k\left(\mathbf{x}_{i_{t}}, \mathbf{x}_{j}\right)<1$} \label{pegasos:condition}
		\STATE $\boldsymbol{\alpha}_{t+1}[i_t] \gets \boldsymbol{\alpha}_t[i_t] + 1$\label{pegasos:increase_alpha}
		
		\ELSE
		\STATE $\boldsymbol{\alpha}_{t+1}[i_t] \gets \boldsymbol{\alpha}_t[i_t]$
		
		\ENDIF

		\ENDFOR
		
		\STATE
		\STATE \textbf{Output:} $\boldsymbol{\alpha}_{T+1}$
	\end{algorithmic}
\end{algorithm}

\subsection{Quantum support vector machines} \label{sec:svm_quantum}
The classical SVM formulation can be straightforwardly adapted to the quantum case by choosing a feature map
\begin{equation*}
			\begin{aligned}
				\psi \colon \R^r & \to \cS(2^q)                                          \\
				\mathbf{x}       & \mapsto \ketbra{\psi(\mathbf{x})}{\psi(\mathbf{x})} \, ,
			\end{aligned}
\end{equation*}
where $\cS(2^q)$ denotes the space of density matrices on $q$ qubits~\cite{Havlicek2019}.
The kernel function is then given by the Hilbert-Schmidt inner product
\begin{equation}\label{eq:quantum_kernel}
		k(\mathbf{x}, \mathbf{y})  = \mathrm{tr} \big[ \ketbra{\psi(\mathbf{y})}{\psi(\mathbf{y})} \, \ketbra{\psi(\mathbf{x})}{\psi(\mathbf{x})} \big]
		                            = |\!\braket{\psi(\mathbf{x})}{\psi(\mathbf{y})}\!|^2 \, ,
\end{equation}
and can be estimated with the help of a quantum computer as illustrated in~\Cref{fig:svm_quantum_kernel_circuit}~\cite{Liu2021,Arne}.  QSVMs are an active research topic~\cite{Rebentrost2014,KBS21,HZ22}.  
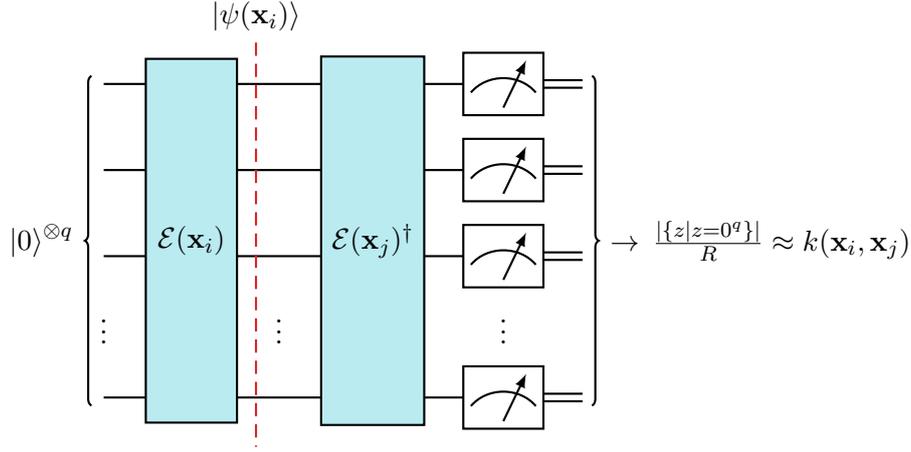
\begin{figure}[!htb]
            \centering
            \begin{quantikz}[row sep=0.2cm]
                \lstick[wires=5]{$\ket{0}^{\otimes q}$} 
                & \gate[5, nwires={4}, style={fill=myCyan!30}][1.2cm]{\cE(\mathbf{x}_i)} \slice[style=myRed]{$\ket{\psi(\mathbf{x}_i)}$} & \qw & \gate[5, nwires={4}, style={fill=myCyan!30}][1.2cm]{\cE(\mathbf{x}_j)^\dagger} & \meter{} & \cw \rstick[wires=5]{$\rightarrow \, \frac{|\{ z | z = 0^q\} |}{R} \approx k(\mathbf{x}_i, \mathbf{x}_j)$}\\
                & & \qw & & \meter{} & \cw\\
                & & \qw & & \meter{} & \cw\\
                \vdots & & \vdots & & \vdots &\\
                & & \qw & & \meter{} & \cw
            \end{quantikz}
            \caption{\textbf{Quantum kernel estimation~\textnormal{\cite{Liu2021,Arne}}:} Let $\cE(\mathbf{x}_i)$ denote a parametrized unitary fixed by the datum $\mathbf{x}_i$, which defines the feature map $\ket{\psi(\mathbf{x}_i)}=\cE(\mathbf{x}_i) \ket{0}^{\otimes q}$. By preparing the state $\cE(\mathbf{x}_j)^\dagger \cE(\mathbf{x}_i)\ket{0}^{\otimes q}$ and then measuring all of the qubits in the computational basis, a bit string $z \in \{0, 1\}^q$ is determined. When this process is repeated $R$-times, the frequency of the all zero outcome approximates the kernel value $k(\mathbf{x}_i, \mathbf{x}_j)$ in~\cref{eq:quantum_kernel}.}
            \label{fig:svm_quantum_kernel_circuit}
        \end{figure}

The major difference between the classically computed kernel in~\Cref{eq:class_kernel} and the quantum one in~\Cref{eq:quantum_kernel} is that the latter expression can only be evaluated approximately with a finite number of measurement shots, due to the probabilistic nature of quantum mechanics. Hence, to prove a complexity statement for QSVMs, it is crucial to understand the robustness of the primal and dual optimization problems with respect to noisy kernel evaluations. A detailed analysis is included in~\Cref{sec_complexity}. 

Note that while such noisy kernel evaluations are necessary for both training and prediction, the complexity of the prediction step is $\cO(S^2 / \varepsilon^2)$ \cite{Arne} and subdominant compared to training (see \cref{tab:qnn_qsvm_scaling}) as $S \leq M$ and often even $S \ll M$ for the number of support vectors $S$ \cite{Burges1998}. 
Therefore, our present analysis is restricted to fitting the support vector machines to a training set. 

\subsection{Approximate quantum support vector machines} 
Compared to the analytically tractable QSVM formulations introduced in~\Cref{sec:svm_quantum}, the  heuristic model we denote \emph{approximate QSVM} defined in the following, potentially offers favorable computational complexity. The main idea is that, in addition to the feature map encoding the data $\ket{\psi(\mathbf{x})} = \cE(\mathbf{x}) \ket{0}^{\otimes q}$ analogously to the QSVM, we implement a variational unitary $\cW (\theta)$, whose parameters $\theta \in \R^d$ are trained to solve the classification problem. \Cref{fig:qnn_circuit} shows a quantum circuit implementing this architecture. 

\begin{figure}[!htb]
    \centering
    \begin{quantikz}[row sep=0.2cm]
			\lstick[wires=5]{$\ket{0}^{\otimes q}$} 
			& \gate[5, nwires={4}, style={fill=myCyan!30}][1.2cm]{\cE(\mathbf{x})} \slice[style=myRed]{$\ket{\psi(\mathbf{x})}$} & \qw & \gate[5, nwires={4}, style={fill=myViolet!30}][2cm]{\cW(\theta)} & \meter{} & \cw \rstick[wires=5]{$\rightarrow h_\theta(\mathbf{x})$}\\
			& & \qw & & \meter{} & \cw\\
			& & \qw & & \meter{} & \cw\\
			\vdots & & \vdots & & \vdots &\\
			& & \qw & & \meter{} & \cw
		\end{quantikz}
    \caption{\textbf{Approximate QSVM~\textnormal{\cite{Havlicek2019,Arne}}:} The classical datum $\mathbf{x}$ is encoded with a feature map circuit $\cE(\mathbf{x})$ analogous to the one in~\Cref{fig:svm_quantum_kernel_circuit}. However, the resulting state is then acted upon by a variational circuit $\cW(\theta)$, after which measurement in the computational basis is performed to determine the expectation value in \Cref{eq_h_theta}.}
    \label{fig:qnn_circuit}
\end{figure}
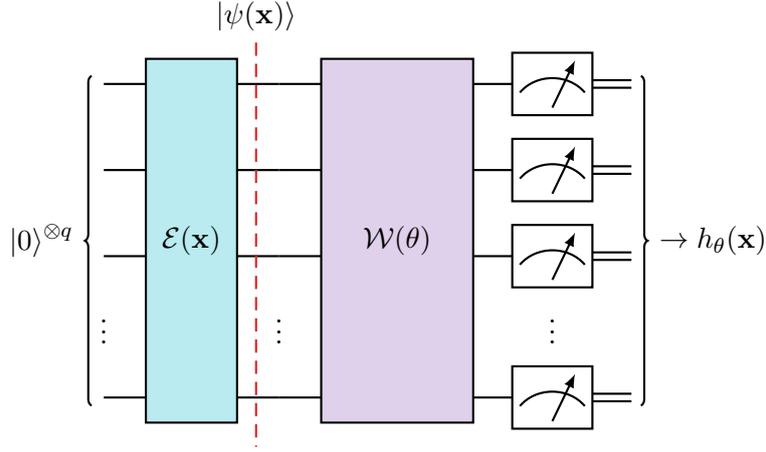
We define the decision function as the expectation value of the $q$-fold $Z$ Pauli-operator~\cite{Havlicek2019}
\begin{equation} \label{eq_h_theta}
	h_\theta(\mathbf{x}) \coloneqq \bra{\psi(\mathbf{x})} \cW(\theta)^\dagger \; Z^{\otimes q} \; \cW(\theta) \ket{\psi({\mathbf{x}})} \in [-1, +1] \, ,
\end{equation}
which together with the introduction of a learnable bias $b \in \R$ is used to classify the data according to
\begin{equation}
    \label{eq:qnn_c}
	c_{\theta}(\mathbf{x}) \coloneqq \text{sign}\left[h_{\theta}(\mathbf{x}) + b\right].
\end{equation}
Given a loss function $\cL$, we use a classical optimization algorithm to minimize the empirical risk on the training data $X$
\begin{equation*}
    \hat{\cR}(c_{\theta}, X) = \frac{1}{M} \sum_{i=1}^M \cL\big(y_i, c_{\theta}(\mathbf{x}_i)\big)\, .
\end{equation*}
The resulting model implements a linear decision boundary in the feature space accessed through the feature map circuit $\cE$ and in this sense approximates a QSVM employing the same feature map (see~\cite{Havlicek2019} and~\cite[Section 2.3.1]{Arne}). 
Note that the model above is just one possible definition to perform binary classification with the circuit depicted in~\Cref{fig:qnn_circuit}. One promising alternative is considering a local instead of a global observable. For both cases the model has to be chosen carefully to avoid unfortunate properties such as Barren plateaus~\cite{marcoBarren21}.

The approximate QSVM model introduced above are also referred to as \emph{quantum neural networks} \cite{Abbas2020a, Belis2021} or as \emph{variational quantum circuits}~\cite{Havlicek2019,review21} in the literature.
\section{Analytical complexity} \label{sec_complexity}
Having introduced the models and associated optimization problems considered in this work, we now turn to analyzing their complexity. The goal of this section is to justify the analytical part of \Cref{tab:qnn_qsvm_scaling}. For the scaling of the dual approach we are able to provide a mathematical proof. The scaling of the primal approach using \Pegasos~can also be analytically derived as long as an additional assumption on the data and feature map is fulfilled. While this assumption is not mathematically proven, we provide empirical evidence to support it. 

\subsection{Dual optimization} \label{sec:complexity_dual}
To solve the dual problem~\eqref{opt:svm_soft_margin_feature_map_dual}, the full kernel matrix $K$, whose elements are given as expectation values which are calculated on a quantum computer (see~\Cref{fig:svm_quantum_kernel_circuit}), needs to be evaluated. As we can only perform a finite number of measurement shots in practice, the estimated kernel entries will always be affected by statistical errors, even when a fault tolerant quantum computer is employed. More precisely, we approximate each kernel entry by taking the sample mean
\begin{align}\label{eq:kernel_sample_mean}
	k_R(\mathbf{x}_i, \mathbf{x}_j) = \frac{1}{R} \sum_{l=1}^R \hat{k}_{ij}^{(l)},
\end{align}
over $R$ i.i.d.~realizations of the random variable $\hat{k}_{ij}^{(l)}$, which is defined to be equal to 1 if the original all zero state $\smash{\ket{0}^{\otimes q}}$ is recovered in the measurement, and equal to 0 for all other measurement outcomes. These entries define an approximate kernel matrix $K_R$, which in the limit $R \to \infty$ converges to the true kernel matrix $K$ by the law of large numbers.
To prove a complexity statement for the dual approach, we need to quantify the speed of this convergence.

A detailed analysis which is shifted to~\Cref{app:latala} for improved readability shows that the expected error on the kernel when using $R$ samples per entry scales as
\begin{equation}\label{eq:latala_result}
 \E \left[ \norm{K_R - K}_2 \right] = \cO\left( \sqrt{\frac{M}{R}} \right).
\end{equation}
Note that throughout this work, the notation $\norm{\cdot}$ denotes the Euclidean vector norm and $\norm{\cdot}_2$ the operator norm it induces. 

In a next step, we analyze how this error propagates through the optimization problem~\Cref{opt:svm_soft_margin_feature_map_dual} to the decision function \Cref{eq:class_qsvm_dual}. To make this step mathematically precise we need to employ an assumption on our kernel and data set.
\begin{assumption} \label{ass_noisy_halfspace}
Our kernel and data set satisfy the noisy halfspace learning assumption defined in \cite[Lemma 14]{Liu2021}.
\end{assumption}
Note that Assumption~\ref{ass_noisy_halfspace} is necessary to rigorously analyze the error propagation as done in~\Cref{app:daniel}. Furthermore, as shown in~\cite{Liu2021} the assumption is reasonable for some kernels. 

Denote by $\alpha_i^\star$ and $\alpha_{R,i}^\star$ the solution to the optimization problem utilizing an exact and noisy kernel matrix respectively. These solutions then lead to the terms 
\begin{equation*}
	h(\hat{\mathbf{x}}) \coloneqq \sum_{i=1}^M \alpha_i^\star y_i \, k(\hat{\mathbf{x}}, \mathbf{x}_i) \qquad \textnormal{and} \qquad
	h_R(\hat{\mathbf{x}}) \coloneqq \sum_{i=1}^M \alpha_{R,i}^\star y_i \, k_R(\hat{\mathbf{x}}, \mathbf{x}_i)\,,
\end{equation*}
inside the sign operator in the classification function \eqref{eq:class_qsvm_dual}, where $\hat{\mathbf{x}}$ is an arbitrary feature vector. We also refer to $h(\mathbf{x})$ as the decision function. We find that
\begin{equation*}
    \left|h(\hat{\mathbf{x}}) - h_R(\hat{\mathbf{x}})\right| = \cO \left(\frac{M^{4/3}}{\sqrt{R}}\right).
\end{equation*}
Thus, in order to achieve an error of $\left|h(\hat{\mathbf{x}}) - h_R(\hat{\mathbf{x}})\right| \leq \varepsilon$, we need a total of
\begin{equation}
    \label{eq:dual_result}
	R_{\textnormal{tot}} = \cO\left(\frac{M^{4.67}}{\varepsilon^2}\right) 
\end{equation}
measurement shots, as the full symmetric kernel matrix consists of $\cO(M^2)$ independent kernel entries. Once the approximated kernel matrix $K_R$ has been calculated, the corresponding quadratic problem can be solved on a classical computer in time $\cO\left(M^{7/2}\log(1/\eps)\right)$\footnote{Quadratic programs are a special case of second order cone problems which can be solved in $\cO\left(M^{7/2}\log(1/\eps)\right)$ using interior point methods~\cite{Boyd2004}.}. Hence, estimating the kernel matrix is the computationally dominant part of the algorithm and the result in~\Cref{eq:dual_result} defines the overall complexity for the dual approach.

\subsection{Primal optimization via \Pegasos}
\label{sec:complexity_pegasos}
In order to analyze the computational complexity of solving the primal optimization problem~\eqref{opt:svm_soft_margin_primal_hinge_loss}, let $f(\mathbf{w})$ denote its objective function. 
Classically, where the kernels are known exactly, the runtime of the kernelized \Pegasos-algorithm scales as $\cO(M/\delta)$, where 
\begin{equation}
    \label{def:delta}
    \delta \coloneqq \left|f(\mathbf{w^\star}) - f(\mathbf{w}^P)\right|
\end{equation}
is the difference in the objective between the true optimizer $\mathbf{w^\star}$ and the result of \Pegasos\, $\mathbf{w}^P$~\cite{Shalev-Shwartz2011a}. When using quantum kernels, however, we have to consider the inherent finite sampling noise afflicting their evaluation.
In order to carry out the complexity analysis for such noisy kernels, we next impose an assumption on Algorithm~\ref{algo:kernelized_pegasos} for which we have empirical evidence as provided in~\Cref{sec:experiments_pegasos}.
\begin{assumption} \label{ass_Pegasos_noise}
The convergence of Algorithm~\ref{algo:kernelized_pegasos} is unaffected if the sum in line~\ref{pegasos:condition} of Algorithm~\ref{algo:kernelized_pegasos} is only $\delta$-accurate for a sufficiently small $\delta>0$.
\end{assumption}

In other words, this means that above a certain approximation accuracy, Algorithm~\ref{algo:kernelized_pegasos} is robust to the statistical uncertainty inherent to quantum kernels which have been evaluated using a finite number of measurement shots.

Thus, we now analyze how many measurement shots are required for the sum in line \ref{pegasos:condition} to be $\delta$-accurate in terms of the standard deviation. First, note that initially $\boldsymbol{\alpha}_1[j] = 0$ for all $j$ and in every iteration at most one component of $\boldsymbol{\alpha}$ becomes non-zero (in line \ref{pegasos:increase_alpha}). At the $t$-th iteration, the sum therefore contains at most $t$-terms. In order for the sum to be $\delta$-accurate, every individual term should thus be $(\delta/\sqrt{t})$-accurate, so we need $t/\delta^2$ measurement shots per term and $t^2/\delta^2$ shots to evaluate the whole sum. Bearing in mind that the total number of iterations $T$ is bounded by $T=\cO(1/\delta)$ \cite{Shalev-Shwartz2011a}, this results in a total of $\cO(T^3/\delta^2) = \cO(1/\delta^5)$ quantum circuit evaluations to train a QSVM with \Pegasos. 
Alternatively, the number of non-zero terms in the sum conditioned on in line~\ref{pegasos:condition} can be bounded by $M$, leading to a total of $\cO(T M^2/\delta^2)= \cO(M^2/\delta^3)$ measurement shots. Combining these results implies $R_{\textnormal{tot}} = \cO(\min \{ M^2/\delta^3, \, 1/\delta^5 \})$.

To directly compare this scaling with~\eqref{eq:dual_result}, we need to clarify the relation between the error
\begin{equation}
\label{def:epsilon}
    \varepsilon \coloneqq \max_{\hat{\mathbf{x}} \in X}|h(\hat{\mathbf{x}}) - h_P(\hat{\mathbf{x}})|
\end{equation}
and $\delta$ defined in~\eqref{def:delta}, where $h$ is the ideal decision function and $h_P$ the decision function resulting from the noisy \Pegasos\, algorithm. As the SVM optimization problem is $\lambda$-strongly convex for $\lambda$ the regularization constant, this connection can be derived as
\begin{equation}     \label{eq:epsilon-delta}
    \varepsilon \leq \sqrt{\frac{2\delta}{\lambda}}\, ,
\end{equation}
see~\Cref{sec:appendix_pegasos_delta_eps}.

Combining these results, the number of measurement shots required to achieve $\varepsilon$-accurate decision functions is bounded by
\begin{equation*}
 R_{\textnormal{tot}} = \cO\left(\min \Big \{ \frac{M^2}{\lambda^3\varepsilon^6}, \, \frac{1}{\lambda^5\varepsilon^{10}} \Big \} \right).
\end{equation*}


\section{Empirical complexity}\label{sec:empirical_scaling}
In this section we provide experiments justifying the empirical part of~\Cref{tab:qnn_qsvm_scaling}. In~\Cref{sec:training_data}, we define the training data used in the following experiments. In~\Cref{app:dual_scaling,app:pegasos_scaling}, we respectively show that the analytical complexity for the dual and primal methods are confirmed in our experiments as close to being tight. Additionally, we provide an empirical scaling for the approximate QSVMs in~\Cref{sec_approximate_QSVM}. \Cref{tab:empirical_scaling} includes the results of the experiments performed in this section for the two settings described in \Cref{sec:training_data}.
\begin{table}[!htb]
\centering

\bgroup
\def\arraystretch{1.5}

    \begin{tabular}{c | c | c | c}
         & dual   & primal (\Pegasos) & approximate QSVM \\ \hline
		separable data &	 $\cO(M^{4.8 \pm 0.4}/\varepsilon^{2})$ & $\cO(1/\varepsilon^{8.3 \pm 1.6})$  & 	$\cO(1/\varepsilon^{2.9 \pm 0.3})$ 		 \\
		overlapping data & $\cO(M^{4.5 \pm 0.3}/\varepsilon^{2})$ & $\cO(1/\varepsilon^{9.5 \pm 1.0})$  & 	$\cO(1/\varepsilon^{2.8 \pm 0.3})$ 
			 
    \end{tabular} 
\egroup

		\caption{\textbf{Empirical complexities of QSVMs}, i.e.~the total number of quantum circuit evaluations required to achieve an $\eps$-accurate solution for the decision function resulting from the dual, primal, and approximate QSVM approach. We consider a data set of size $M$.
}
		\label{tab:empirical_scaling}
	\end{table}

\subsection{Training data}\label{sec:training_data}
The training data used throughout the following experiments is artificially generated according to~\Cref{algo:generating_artificial_data} in the Appendix with respect to a fixed feature map of the architecture illustrated in~\Cref{fig:feature_circuit}. Unless declared otherwise, the experiments are run on 8 qubits and the data thus consists of 8 features. In particular, we consider two settings: data sets that are linearly separable in feature space and ones where there is an overlap between the classes. One realization per setting is illustrated in~\Cref{fig:data_sets}.
\begin{figure}[!htb]
    \centering
     \begin{quantikz}
        \lstick{$\ket{0}$} & \gate{H}\gategroup[3,steps=4,style={dashed,rounded corners,fill=myCyan!30, inner xsep=2pt},background,label style={label position=below,anchor=north,yshift=-0.2cm}]{{Repeat 4 times}} & \gate{R_z(\pi \cdot x_1)} & \gate[wires=2]{\begin{array}{c} ZZ \\ (\pi \cdot x_1 \cdot x_2) \end{array}} & \qw & \meter{} & \cw \\
        \lstick{$\ket{0}$} & \gate{H} & \gate{R_z(\pi \cdot x_2)} & &\gate[wires=2]{\begin{array}{c} ZZ \\ (\pi \cdot x_2 \cdot x_3) \end{array}}&  \meter{} & \cw \\
        \lstick{$\ket{0}$} & \gate{H} & \gate{R_z(\pi \cdot x_3)} &\qw & &\meter{} & \cw
    \end{quantikz}
    \caption{\textbf{Feature map circuit:} This quantum circuit (scaled to 8 qubits with nearest-neighbour entanglement) is employed as the feature map in the QSVMs throughout \Cref{sec:empirical_scaling}. $H$ denotes a Hadamard-gate, $R_z(\theta)$ a single-qubit rotation about the $z$-axis, and $ZZ(\theta)$ a parametric 2-qubit $Z\otimes Z$ interaction (maximally entangled for $\theta = \pi/2$). The feature components $x_1$ and $x_2$ (where $\mathbf{x} = (x_1, x_2)$ is a classical input datum) provide the angles for the rotations. This feature map was chosen due to its previous use in the literature~\cite{Abbas2020a,Havlicek2019}.}
    \label{fig:feature_circuit}
\end{figure}
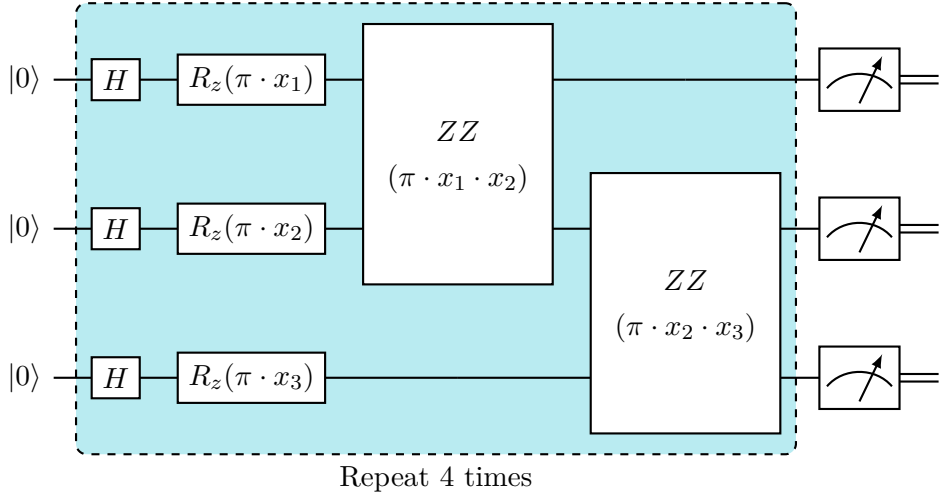

\begin{figure}[!htb]
\centering
\begin{subfigure}{.5\textwidth}
  \centering
  \includegraphics[width=\linewidth]{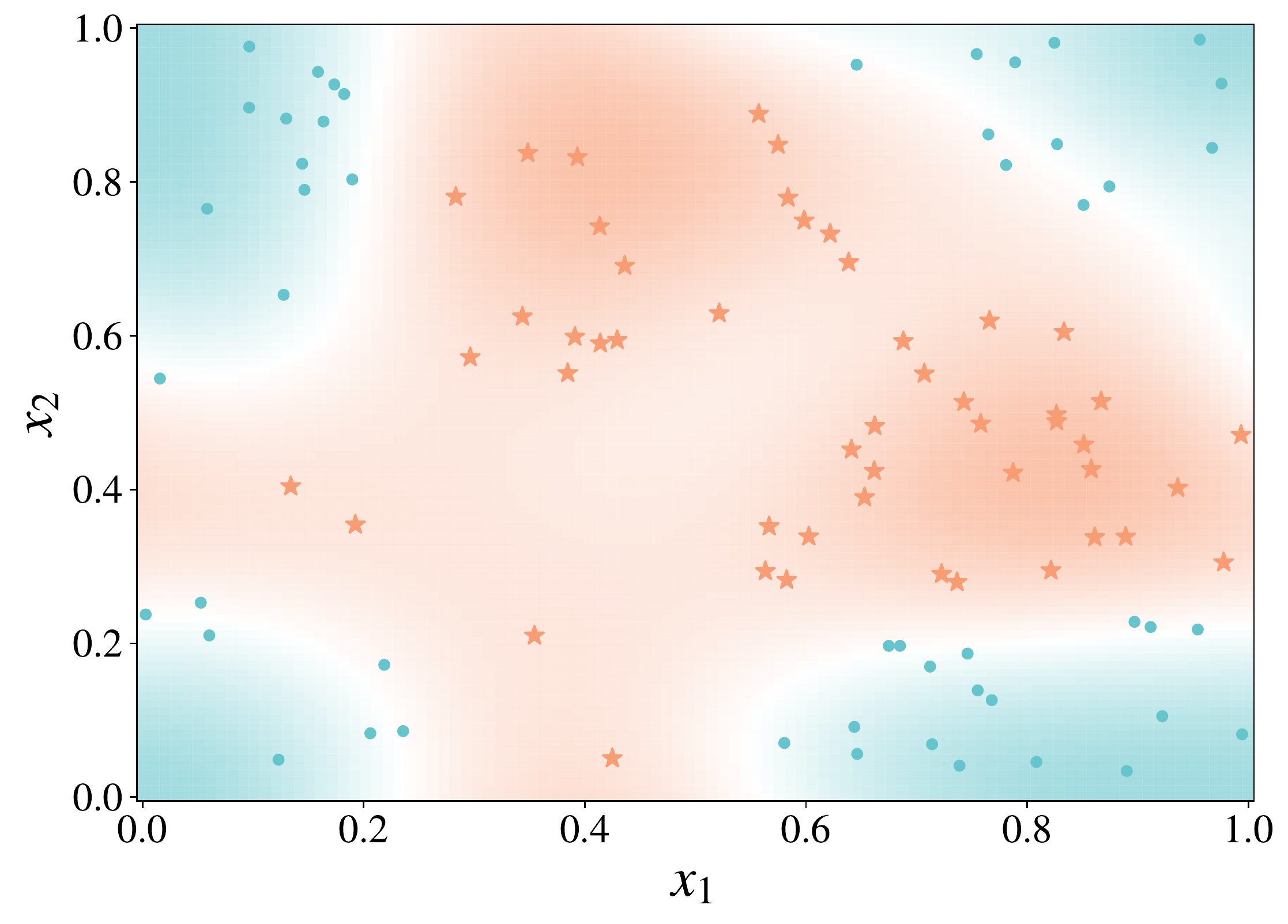}
  \caption{Linearly separable data}

\end{subfigure}%
\begin{subfigure}{.5\textwidth}
  \centering
  \includegraphics[width=\linewidth]{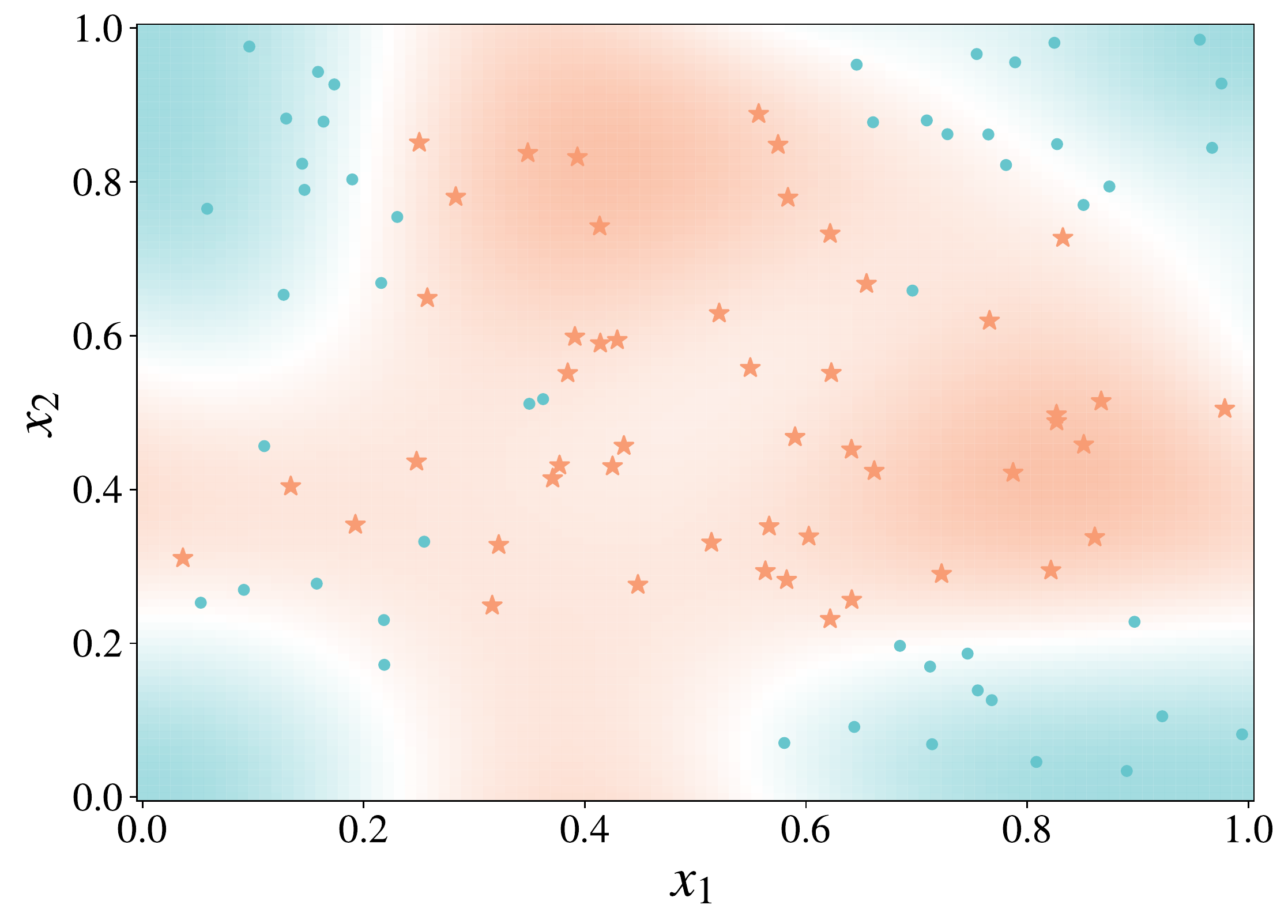}
  \caption{Overlapping data}

\end{subfigure}
\caption{\textbf{Artificial data:} The training data is generated using \Cref{algo:generating_artificial_data} with the feature map from~\Cref{fig:feature_circuit} and a sample size $M=100$. The linearly separable data is achieved by setting the margin positive ($\mu=0.1$ in this case), while the margin for the overlapping data is negative ($\mu=-0.1$). The colouring in the background corresponds to the fixed classifier used to generate the data (note that this is not necessarily the ideal classifier for the generated data points). To allow visualisation, the data displayed here is generated for 2 qubits according to the same algorithm.}
\label{fig:data_sets}
\end{figure}
\newpage
\subsection{Dual optimization}
\label{app:dual_scaling}
In addition to the complexity-theoretic scaling derived in~\Cref{sec:complexity_dual}, we provide an empirical scaling for the dual optimization problem based on numerical experiments. Using shot based noisy kernels, we first determine the $\eps$-dependence of the runtime. In a separate experiment, we analyze the $M$-dependence, which the theory indicates to pose the main limitation of the method. 


For the $\eps$-dependence, we fix the data size to $M = 256$. First, the exact kernel is calculated with a statevector simulator~\cite{qiskit}. The quadratic program~\Cref{opt:svm_l2_matrix_dual} is solved with help of the \texttt{quadprog} python library~\cite{quadprog} in order to find a reference decision function $h_{\infty}(\mathbf{x})$ corresponding to an infinite number of measurement shots. In the next step, we emulate the shot based kernel entries\footnote{Emulating the noisy kernel in this way is equivalent to using a QASM simulator but computationally more efficient.}. We know that for the exact kernel entries $K_{ij}$, the probability of measuring the all $\ket{0}$ outcome is equal to $|\!\braket{\psi(\mathbf{x}_i)}{\psi(\mathbf{x}_j)}\!|^2 = K_{ij}$. We can thus emulate a sample mean of $R$ measurement shots using the binomial distribution $B(R,K_{ij})$. This procedure is repeated for values of $R$ ranging from $10^9$ to $10^{19}$~\footnote{The numbers of shots are chosen so large in order for the condition $||K - K_R|| < \mu$ in~\Cref{th:daniels} to be fulfilled. For smaller $R$, we would need to pre-process $K_R$ to be positive definite in some way, further complicating the analysis.} to find the noisy decision functions $h_{R}(\mathbf{x})$. The error is then defined as
\begin{equation}
    \label{eq:eps_dual_heuristic}
    \eps \coloneqq \max_{\mathbf{x}_i \in X} \left|h_R(\mathbf{x}_i) - h_{\infty}(\mathbf{x}_i)\right|\,.
\end{equation}
Analyzing the plots in~\Cref{fig:dual_exponent}, we find that the bound $R = \cO(1/\eps^{2})$ predicted in~\Cref{sec:complexity_dual} is tight and in precise agreement with the experiments. 


\begin{figure}[!htb]
\centering
\begin{subfigure}{.5\textwidth}
  \centering
  \includegraphics[width=\linewidth]{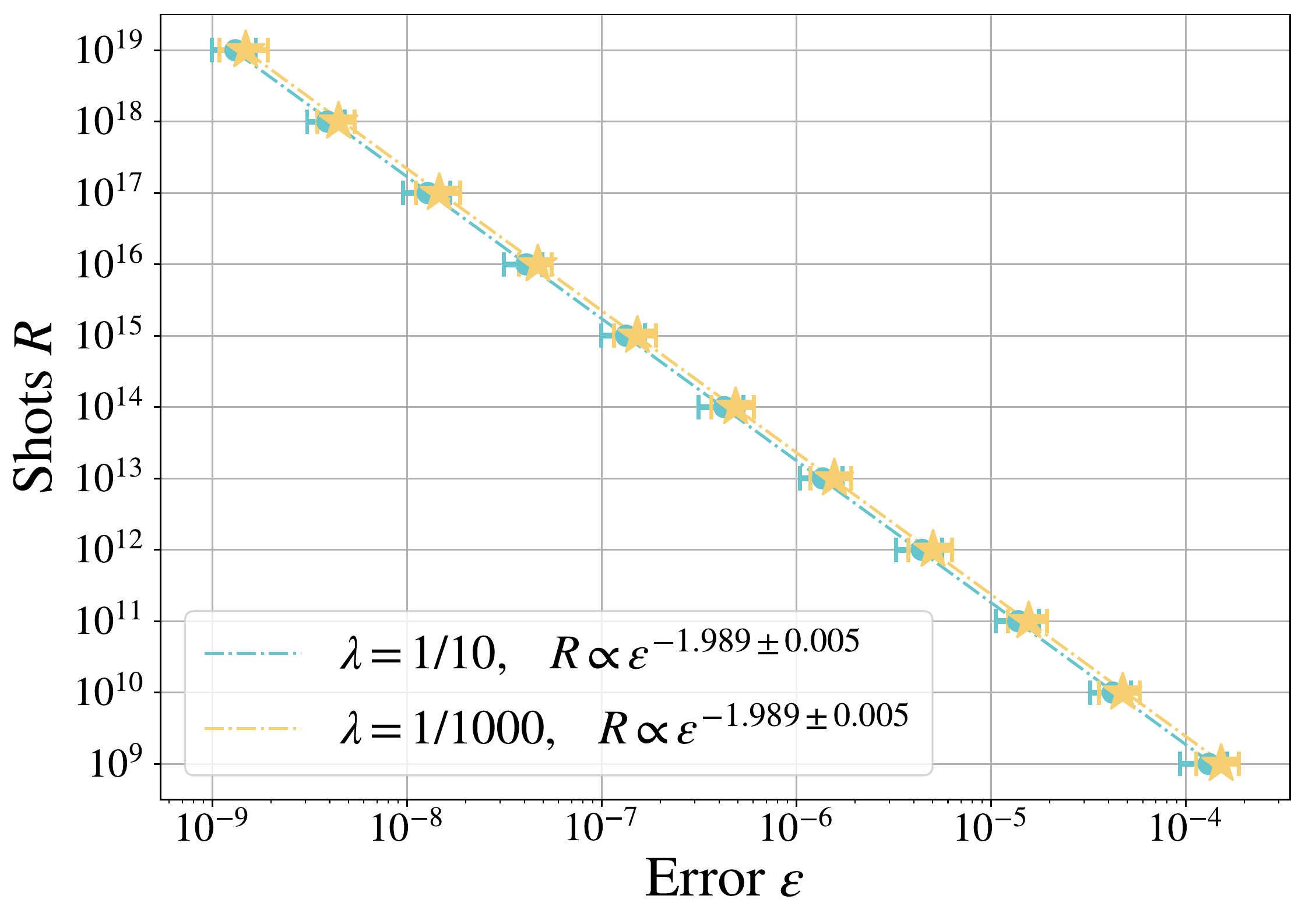}
  \caption{Linearly separable data}

\end{subfigure}%
\begin{subfigure}{.5\textwidth}
  \centering
  \includegraphics[width=\linewidth]{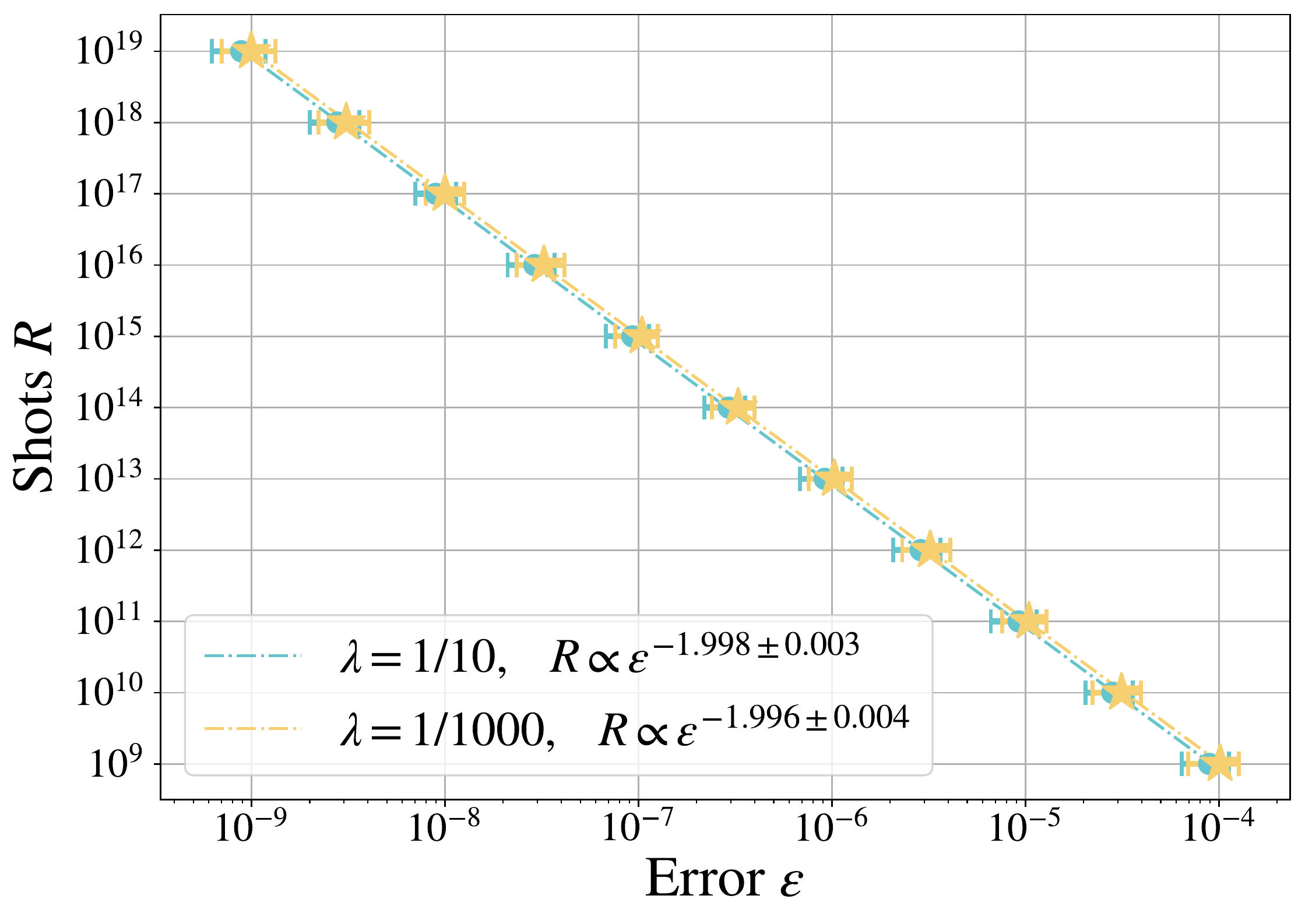}
  \caption{Overlapping data}

\end{subfigure}
\caption{\textbf{$\boldsymbol{\eps}$-scaling of the dual:} Using noisy kernel evaluations drawn from a binomial distribution to simulate the number of shots $R$, the kernel matrix $K_R$ is constructed and the dual optimization problem solved. The number of shots is plotted as a function of $\eps$ on a doubly logarithmic scale, where $\eps$ is calculated according to~\Cref{eq:eps_dual_heuristic}. A linear fit inside the log-log plot is then used to determine the empirical exponent. The experiment is repeated for the different regularization parameters $\lambda = 1/10$ and $\lambda = 1/1000$ and run on $n=100$ different realizations of the training data (see~\Cref{sec:training_data}). The markers shown are the means over the different runs and the horizontal error bars correspond to the interval between the 15.9 and 84.1 percentile.}
\label{fig:dual_exponent}
\end{figure}

In the second experiment, we analyze how the runtime of the dual optimization problem scales with $M$. To this end, the dual problem is solved for different data sizes $M$ and shots per kernel evaluation $R$ to calculate the corresponding error $\eps(M,R)$ according to~\Cref{eq:eps_dual_heuristic}. We then fix some $\eps_0 > 0$ and define $R_{\eps_0}(M)$ as the smallest $R$ such that $\eps(M,R) < \eps_0$, i.e.~$R_{\eps_0}(M)$ is the minimal number of shots needed to achieve an $\eps_0$-accurate solution. In \Cref{fig:dual_M}, a log-log plot of $R_{\textnormal{tot}} = R_{\eps_0}(M)\left[M(M+1)/2\right]$\footnote{The factor $[M(M+1)/2]$ comes from the number of independent entries in the symmetrix $M \times M$ kernel matrix.} as a function of $M$ for different $\eps_0$ is included. From a linear least squares fit in this plot the empirical scaling is determined as $R_{\textnormal{tot}} = \cO(M^{4.8 \pm 0.4})$ for separable data and $R_{\textnormal{tot}} = \cO(M^{4.5 \pm 0.3})$ for overlapping data.

\begin{figure}[!htb]
\centering
\begin{subfigure}{.5\textwidth}
  \centering
  \includegraphics[width=\linewidth]{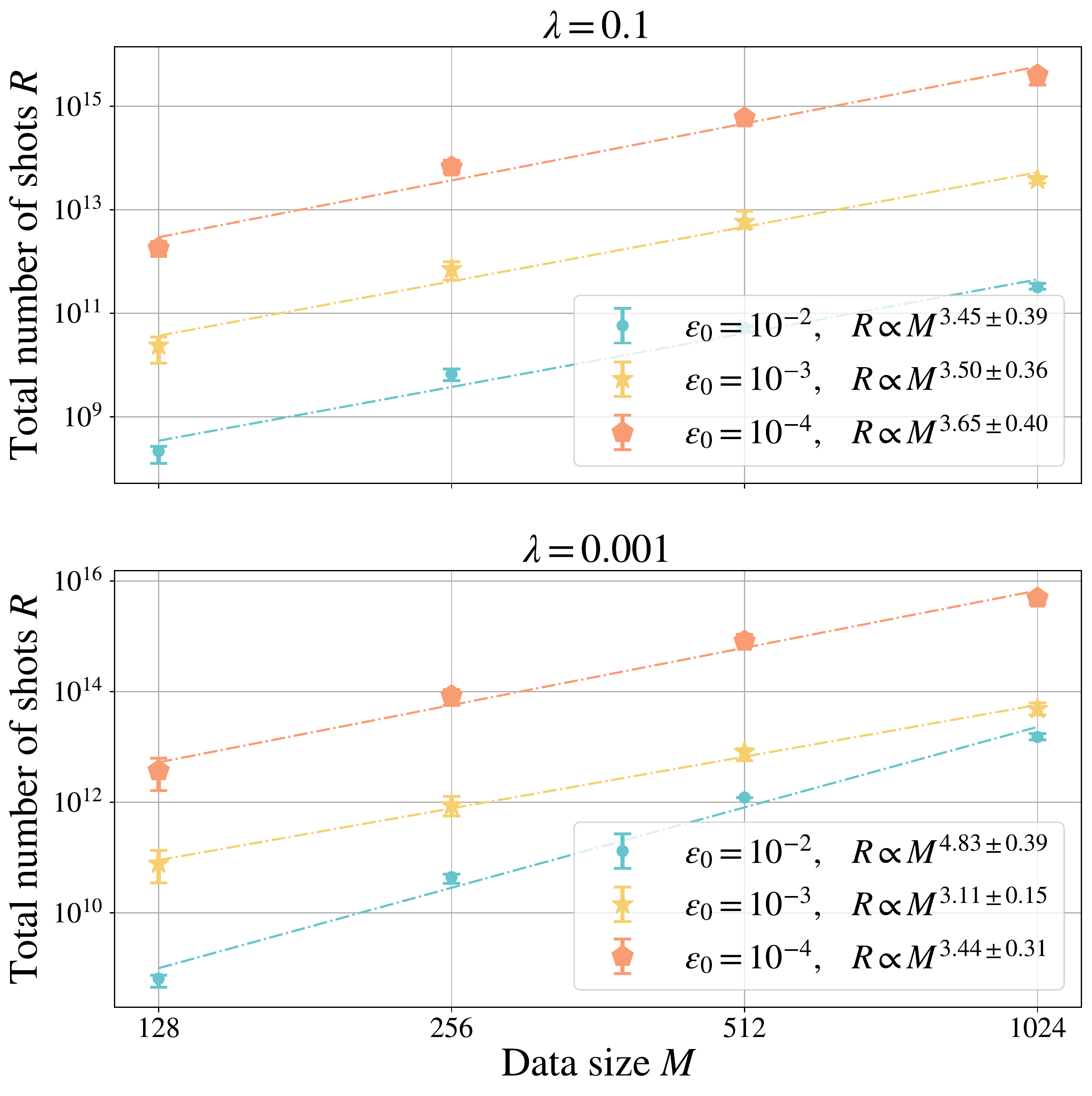}
  \caption{Linearly separable data}

\end{subfigure}%
\begin{subfigure}{.5\textwidth}
  \centering
  \includegraphics[width=\linewidth]{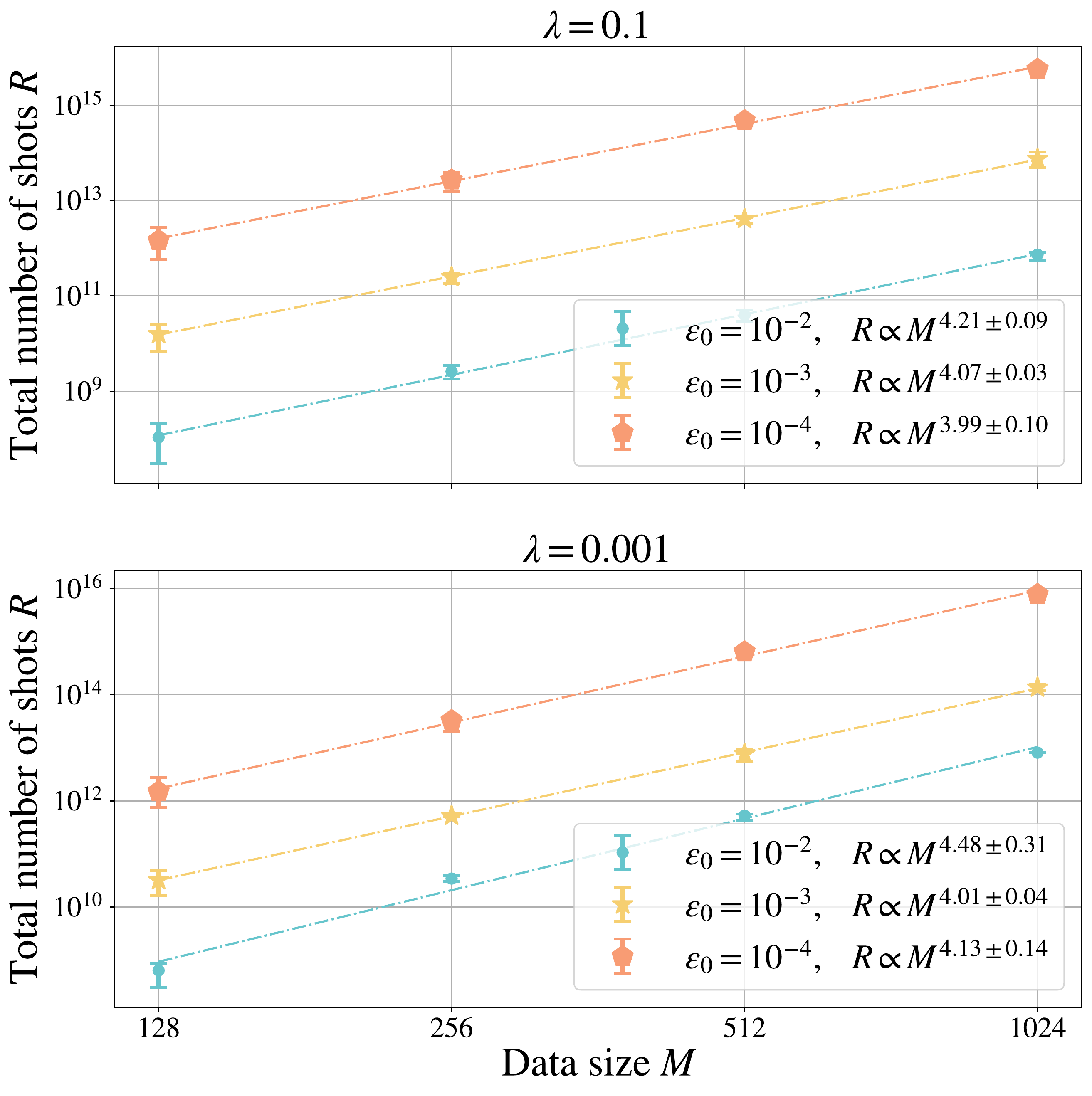}
  \caption{Overlapping data}

\end{subfigure}
\caption{\textbf{$\boldsymbol{M}$-scaling of the dual:} It is plotted how the minimal number of total shots necessary to achieve an $\eps_0$-accurate decision functions depends on $M$. The experiment is repeated for the different regularization parameters $\lambda = 1/10$ and $\lambda = 1/1000$. The markers shown are the means and the error bars indicate the interval between the 15.9 and 84.1 percentile.}
\label{fig:dual_M}
\end{figure}

\subsection{Primal optimization via \Pegasos}
\label{app:pegasos_scaling}
In order to determine the empirical scaling of the \Pegasos~algorithm, we observe how the error $\eps$ of the decision function changes when the number of shots per evaluation is varied. Again employing the feature map and training data described in \Cref{sec:training_data} (with size $M=100$), we first train a reference decision function $h_{\infty}(\mathbf{x})$ by
running \Pegasos~for $T=1000$ iterations with a statevector simulator. This number of steps is sufficient for convergence in our case as can be seen in~\Cref{fig:pegasos_convergence}. In a next step, QASM-simulators~\cite{qiskit} with fixed number of shots per kernel evaluation $R$ are used to run \Pegasos. After every iteration $t$, we calculate the hinge loss 
\begin{equation*}
   \cL_R^t = \frac{1}{M}\sum_{i=1}^M \max \left\lbrace 0, 1 - y_i h_R^t(\mathbf{x}_i)\right\rbrace\,,
\end{equation*} where $h_R^t(\mathbf{x})$ is the decision function after $t$ iterations and $y_i$ the true label. This procedure is repeated for $T$ iterations until convergence, which is defined to be reached when
\begin{equation}
    \left|\cL_R^t - \cL_R^{t-1}\right| < \tau
    \label{eq:empirical_pegasos_convergence}
\end{equation}
for some tolerance which is fixed to $\tau = 10^{-4}$ for this experiment. Finally, the decision function after convergence $h_R^T$ is compared to the reference decision function, defining an error
\begin{equation}
    \label{eq:eps_pegasos_empirical}
    \eps \coloneqq \max_{\mathbf{x}_i \in X} \left|h_R^T(\mathbf{x}_i) - h_{\infty}(\mathbf{x}_i)\right|.
\end{equation}
\Cref{fig:pegasos_exponent} shows the relationship between $\eps$ and $R$. \Cref{fig:pegasos_convergence} further suggests that the number of iterations needed until convergence is independent of the number of shots $R$ per kernel evaluation for sufficiently large R. Thus, we can assume the total number of shots to scale as $R_{\textnormal{tot}} = TR$ for some constant $T > 0$. Under this assumption, the experiments yield an empirical scaling of 
\begin{equation*}
    R_{\textnormal{tot}} = \cO(R) \approx \cO\left(\frac{1}{\eps^{9.5 \pm 1.0}}\right)
\end{equation*}
with respect to $\eps$.

\begin{figure}[!htb]
\centering
\begin{subfigure}{.5\textwidth}
  \centering
  \includegraphics[width=\linewidth]{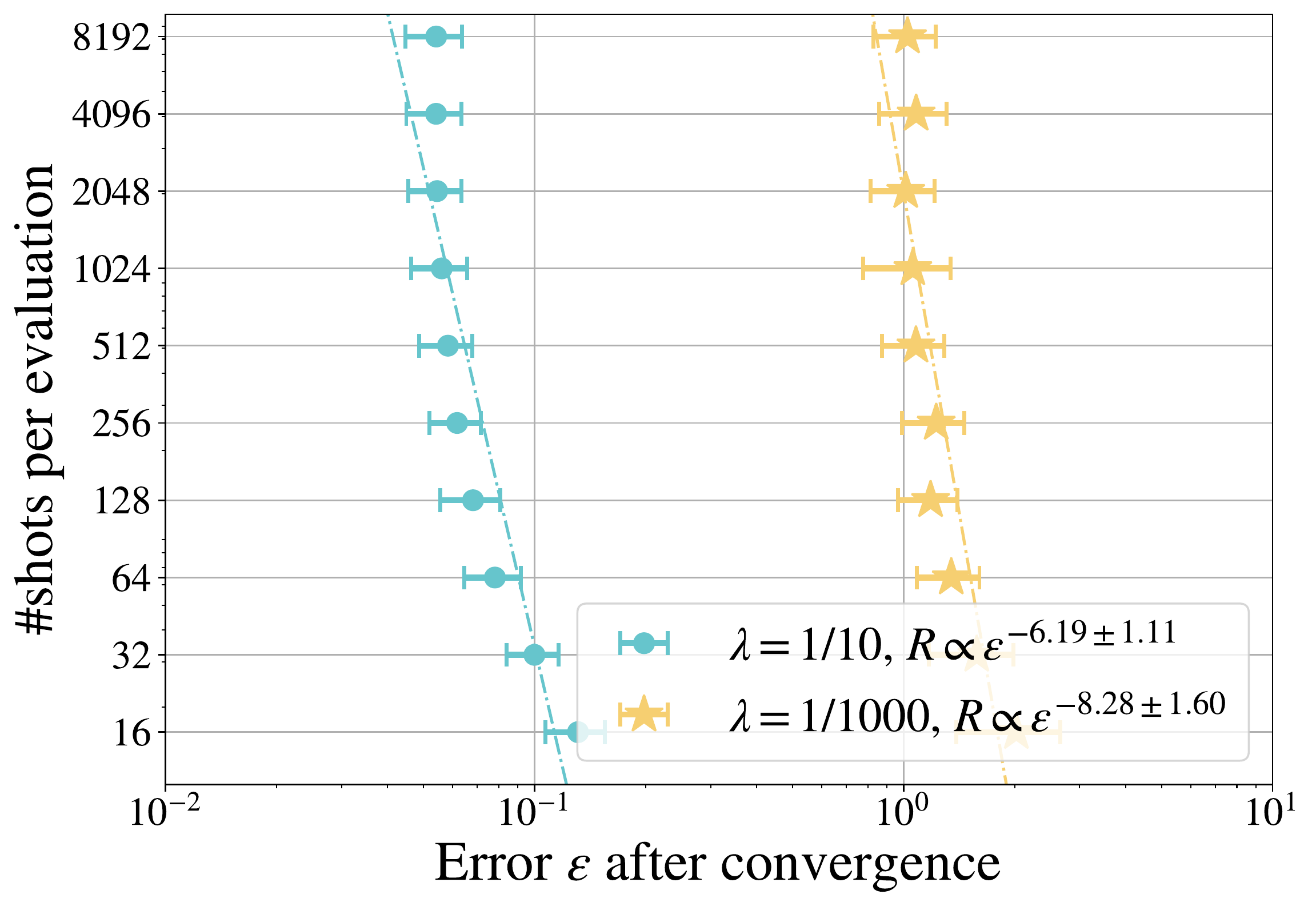}
  \caption{Linearly separable data}

\end{subfigure}%
\begin{subfigure}{.5\textwidth}
  \centering
  \includegraphics[width=\linewidth]{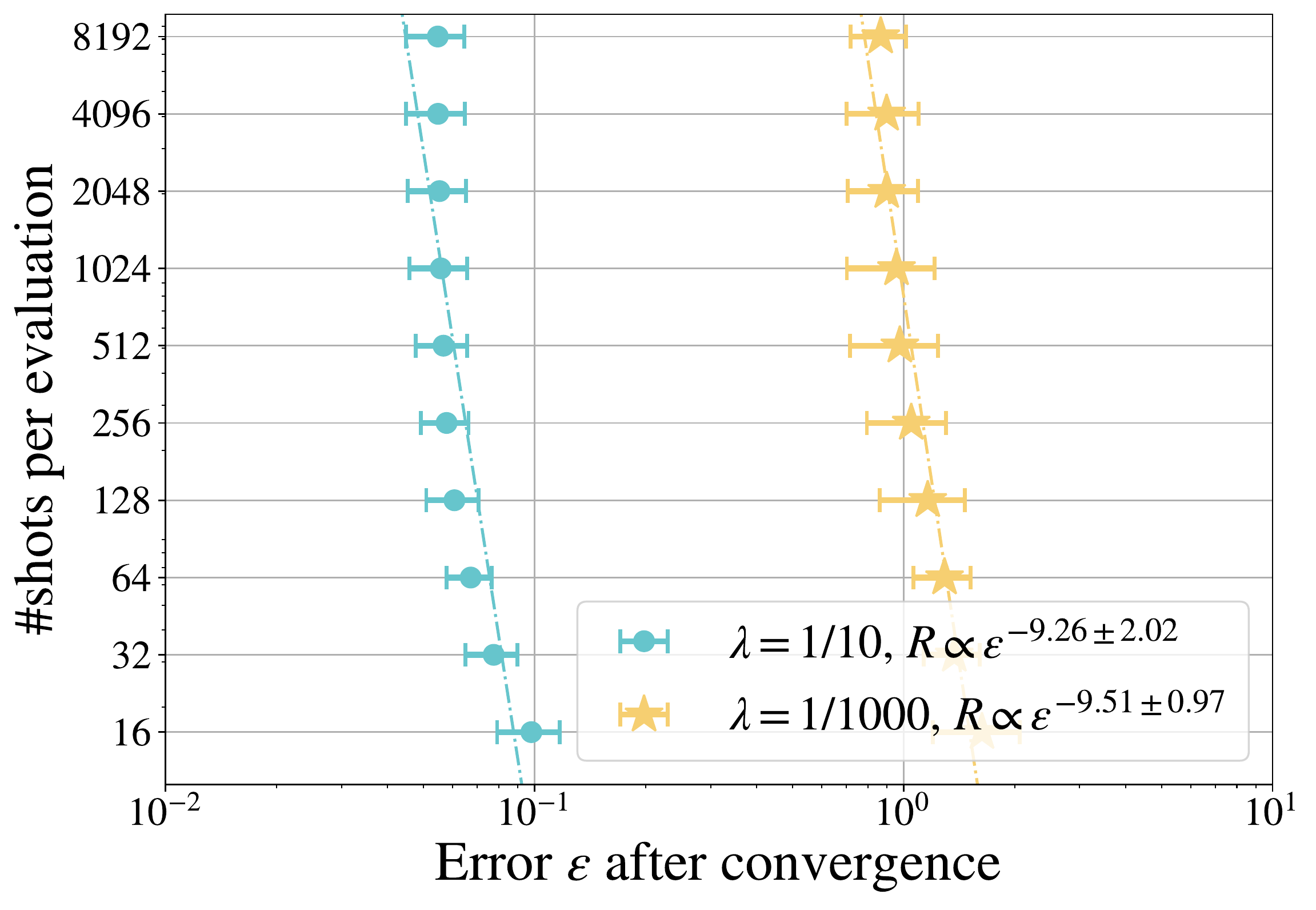}
  \caption{Overlapping data}

\end{subfigure}
\caption{\textbf{$\boldsymbol{\eps}$-scaling of \Pegasos:} Using QASM-simulators with 16 to 8192 shots per evaluation, the \Pegasos~algorithm is run for 750 iterations to ensure convergence. The number of shots is plotted as a function of $\eps$ on a doubly logarithmic scale, where $\eps$ is calculated as defined in \eqref{eq:eps_pegasos_empirical}. A linear fit is used to determine the empirical exponent. The experiment is repeated for the different regularization parameters $\lambda = 1/10$ and $\lambda = 1/1000$. Every data point corresponds to the mean over 50 random runs of the \Pegasos-algorithm and the error bars indicate the standard deviation.}

\label{fig:pegasos_exponent}
\end{figure}

\begin{figure}[!htb]
\centering
\begin{subfigure}{.5\textwidth}
  \centering
  \includegraphics[width=\linewidth]{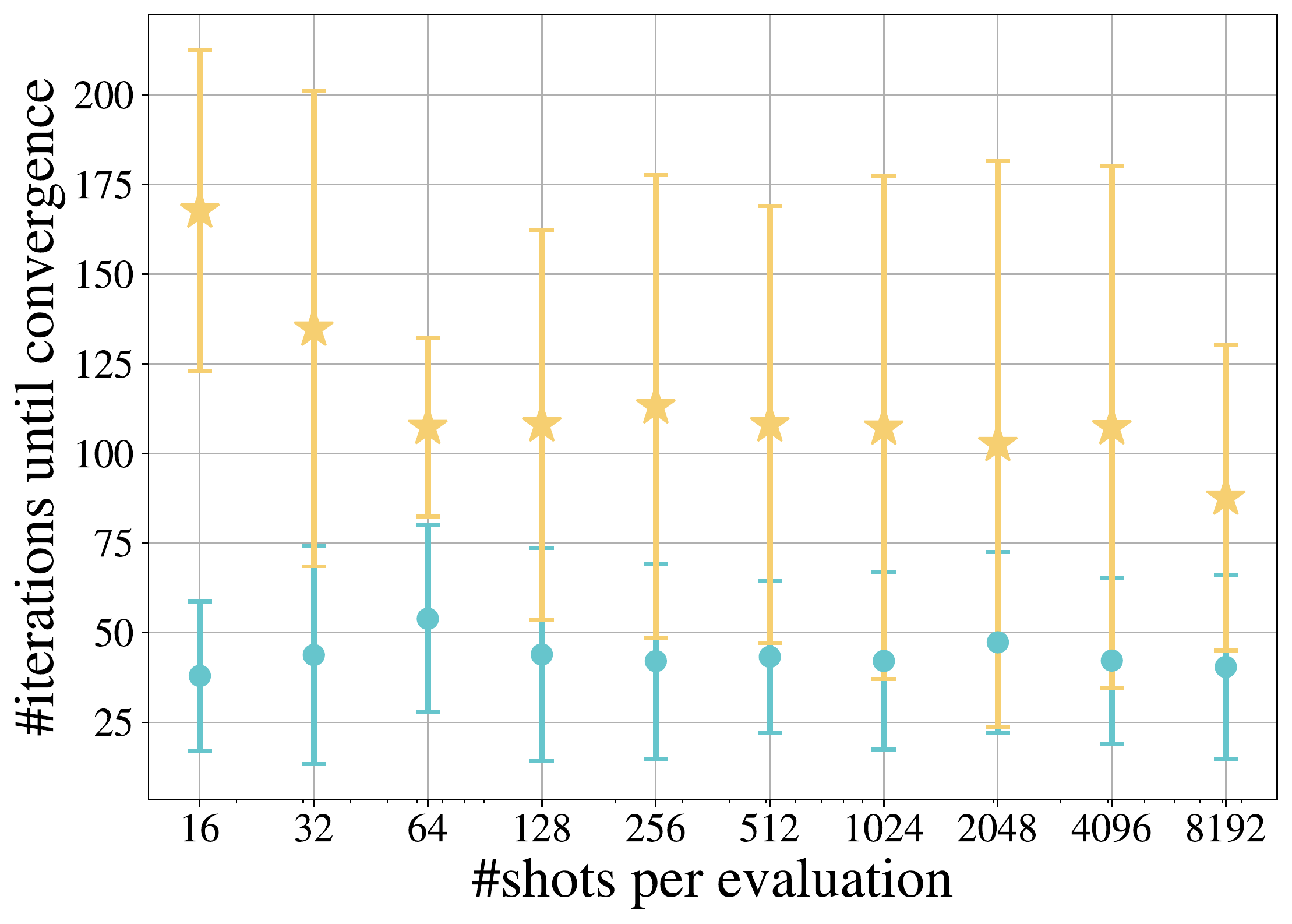}
  \caption{Linearly separable data}

\end{subfigure}%
\begin{subfigure}{.5\textwidth}
  \centering
  \includegraphics[width=\linewidth]{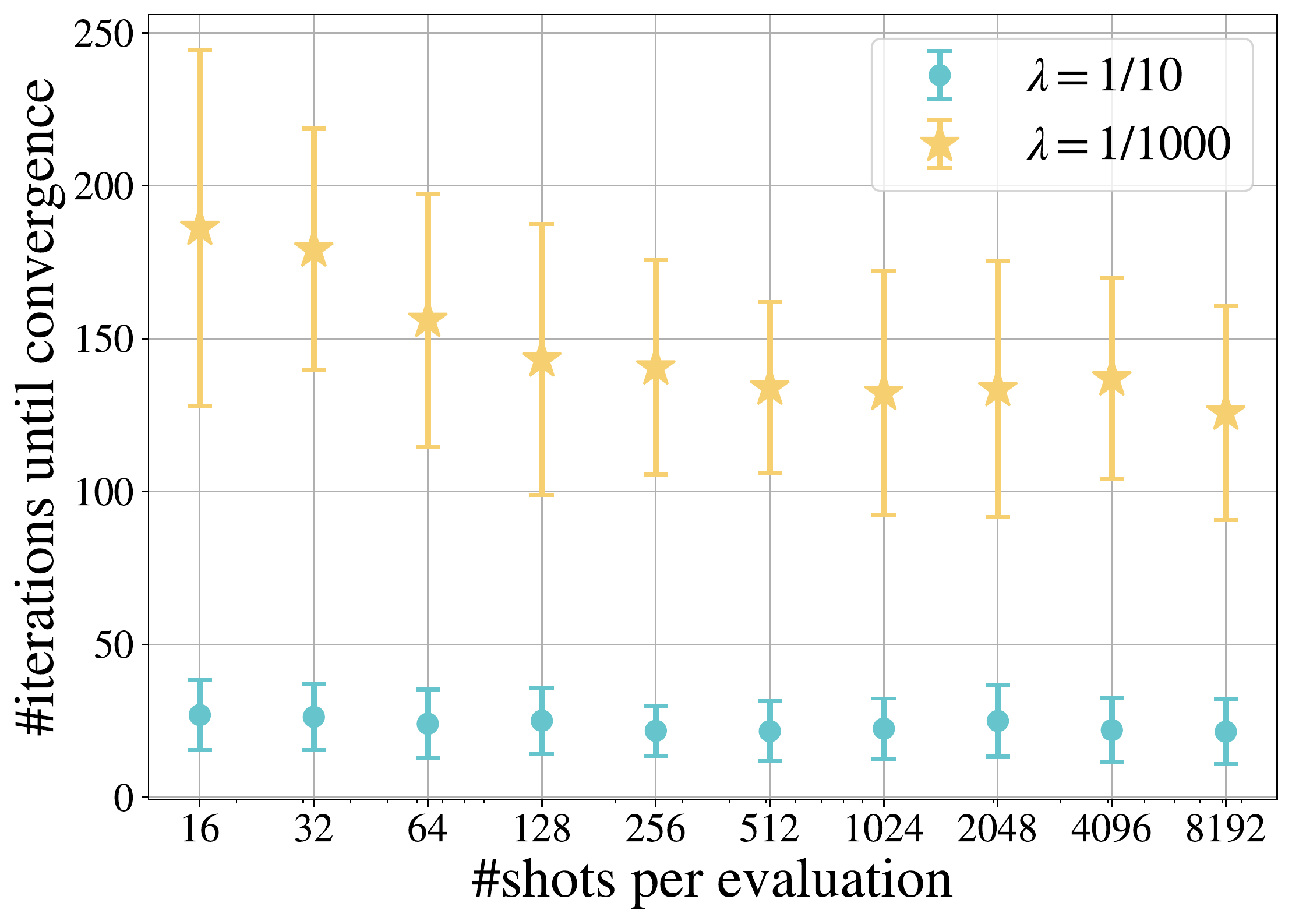}
  \caption{Overlapping data}

\end{subfigure}
\caption{\textbf{Convergence of \Pegasos:} Like in~\Cref{fig:pegasos_exponent}, QASM-simulators ranging from 16 to 8192 shots per evaluation are employed to run \Pegasos\ until convergence as defined in \Cref{eq:empirical_pegasos_convergence}. The number of iterations until convergence is plotted as a function of the number of shots per kernel evaluation. The experiment is repeated for the different regularization parameters $\lambda = 1/10$ and $\lambda = 1/1000$. For every plotted data point, 50 runs are performed using different random seeds for the choice of indices in the \Pegasos~algorithm. The markers shown are the means of the respective steps until convergence and the error bars mark the standard deviation.}
\label{fig:pegasos_convergence}
\end{figure}

\begin{figure}[!htb]
\centering
\begin{subfigure}{.5\textwidth}
  \centering
  \includegraphics[width=\linewidth]{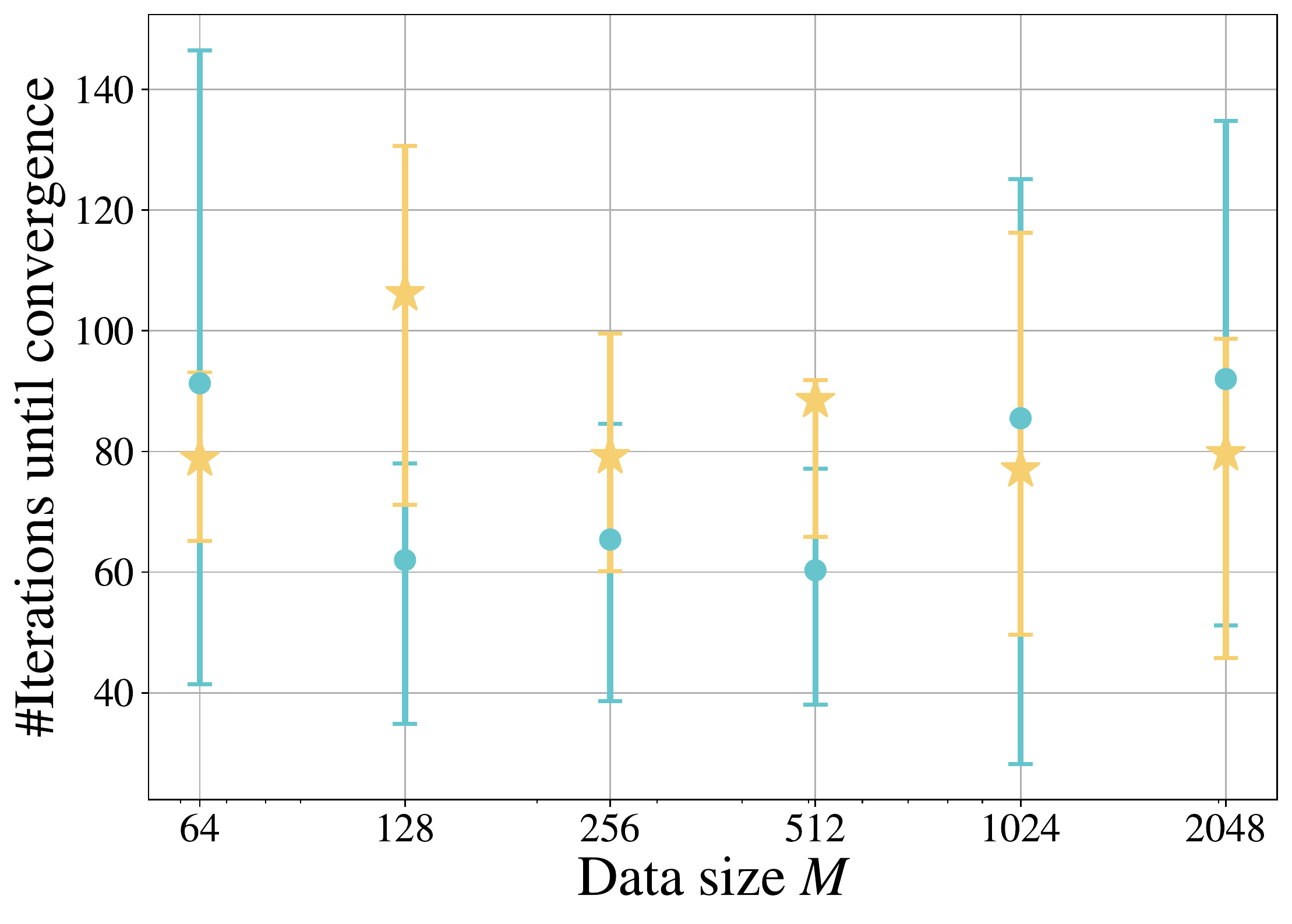}
  \caption{Linearly separable data}

\end{subfigure}%
\begin{subfigure}{.5\textwidth}
  \centering
  \includegraphics[width=\linewidth]{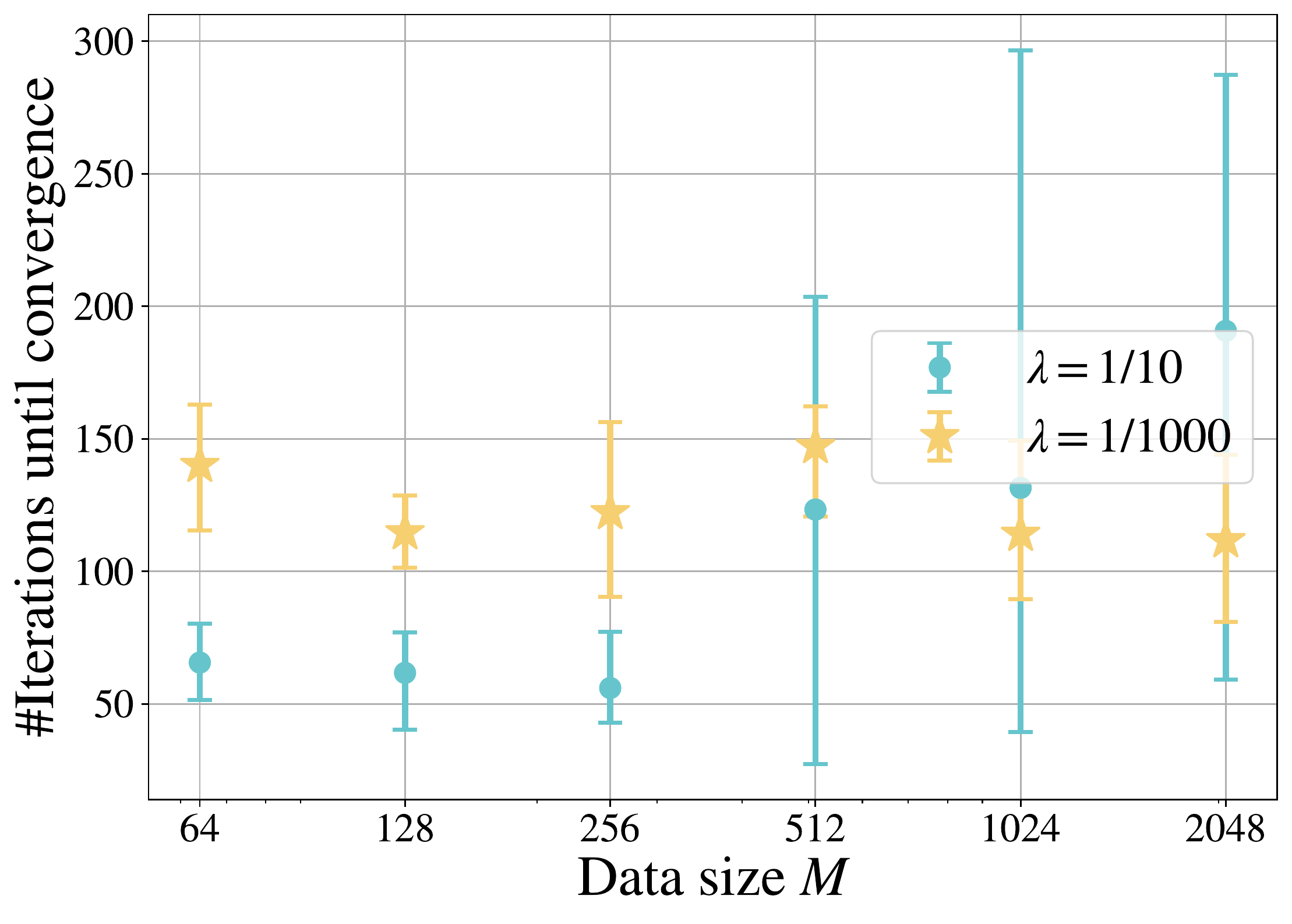}
  \caption{Overlapping data}

\end{subfigure}
\caption{\textbf{$\boldsymbol{M}$-scaling of \Pegasos:} \Pegasos~is applied on the artificial data introduced in~\Cref{sec:training_data} for different training set sizes $M$. The number of iterations until convergence using a statevector simulator is plotted as a function of $M$. The experiment is repeated for the two regularization parameters $\lambda = 1/10$ and $\lambda = 1/1000$. Every data point is the mean over 10 different runs of the probabilistic \Pegasos-algorithm and the error bars mark the interval between the 15.9 and 84.1 percentile.}
\label{fig:pegasos_M}
\end{figure}

In addition to the $\eps$-dependence of the complexity of \Pegasos, we should also remark on the $M$-dependence. \Cref{fig:pegasos_convergence} shows that for large enough $R$, the number of steps needed to converge is independent of $R$. Thus, it suffices to analyze how fast \Pegasos\ converges for varying $M$ when a statevector simulator corresponding to an infinite number of measurement shots is used. \Cref{fig:pegasos_M} shows that there is no discernible correlation between the number of steps until convergence of \Pegasos~and the data size $M$ as expected from the theoretical analysis in~\Cref{sec:complexity_pegasos}.


\subsection{Heuristic training of approximate QSVMs} \label{sec_approximate_QSVM}
Before we tackle the empirical scaling of the non-convex optimization problem that the training of approximate QSVMs poses, we analyze how the expectation value of $h_\theta(\mathbf{x})$ defined in~\eqref{eq_h_theta} is affected by finite sampling statistics. Let $Q \in \{-1,1\}$ denote the random measurement outcome, then $\frac{1}{2}(1 + Q) \sim \textnormal{Bernoulli}(p)$ follows a Bernoulli distribution where $p$ is the probability that $Q =1$. For an i.i.d.\ sample of size $R$, the standard deviation of the sample mean \smash{$Q_R \coloneqq \frac{1}{R}\sum_{i=1}^R Q_i$} is given as \smash{$\sigma_{Q_R} = \sigma_Q/\sqrt{R}$}, where $\sigma_Q$ denotes the standard deviation of the distribution from which the $Q_i$ are sampled. With $\sigma_{\textnormal{Bernoulli}}^2 = p(1 - p) \leq 1/4$ we find $\sigma_Q^2 = 4 \sigma_{\textnormal{Bernoulli}}^2 \leq 1$ and, thus, if we require the standard deviation to be bounded by $\sigma_{Q_R} = \cO(\eps)$ in order to get a confidence interval of width $\cO(\eps)$, we need to perform $R = \cO(1/\eps^2)$ measurement shots per expectation value. Assuming that the training converges in $\cO(1/\eps)$ steps\footnote{Although the theory is not applicable here, we assume a convergence rate of $\cO(1/\eps)$ like first order methods for convex problems~\cite{ref:nesterov-book-04}.}, the expected total number of shots required to fit the model to the training set then scales as
\begin{equation}
    \label{eq:qnn_expected_R}
   R_{\textnormal{tot}}=\cO\left(\frac{1}{\eps^3}\right).
\end{equation}

Because the optimization problem is non-convex, a rigorous proof of the conjectured scaling in~\Cref{eq:qnn_expected_R} remains an unsolved problem. Instead, we perform numerical experiments to determine an empirical scaling. The circuit employed to implement the approximate QSVM is of the form shown in~\Cref{fig:qnn_circuit}, where the unitaries $\cE(\mathbf{x})$ and $\cW(\theta)$ are \href{https://qiskit.org/documentation/stubs/qiskit.circuit.library.ZZFeatureMap.html}{\texttt{ZZFeatureMap}} and \href{https://qiskit.org/documentation/stubs/qiskit.circuit.library.RealAmplitudes.html#qiskit.circuit.library.RealAmplitudes}{\texttt{RealAmplitudes}} circuits, respectively, as provided by the Qiskit Circuit Library~\cite{qiskit}. The experiment is performed on both separable and overlapping data generated as outlined in~\Cref{sec:training_data} and a size of $M=100$. To train the parameters $\theta$, we minimize the \href{https://qiskit.org/documentation/machine-learning/stubs/qiskit_machine_learning.utils.loss_functions.CrossEntropyLoss.html}{\texttt{cross-entropy loss}} using a gradient based method. Calculating the full gradient scales linearly with both the training set size $M$ and the number of parameters $d$. To avoid this, we use SPSA~\cite{SPSA}, which ensures that the scaling is independent of $d$. In addition, we use stochastic gradient descent to approximate the gradient on a batch of 5 data points instead of the whole training set, removing the $M$-dependence. \Cref{fig:approx_d_and_M} includes empirical evidence in support of this claim. 

\begin{figure}[!htb]
\centering
\begin{subfigure}{.5\textwidth}
  \centering
  \includegraphics[width=\linewidth]{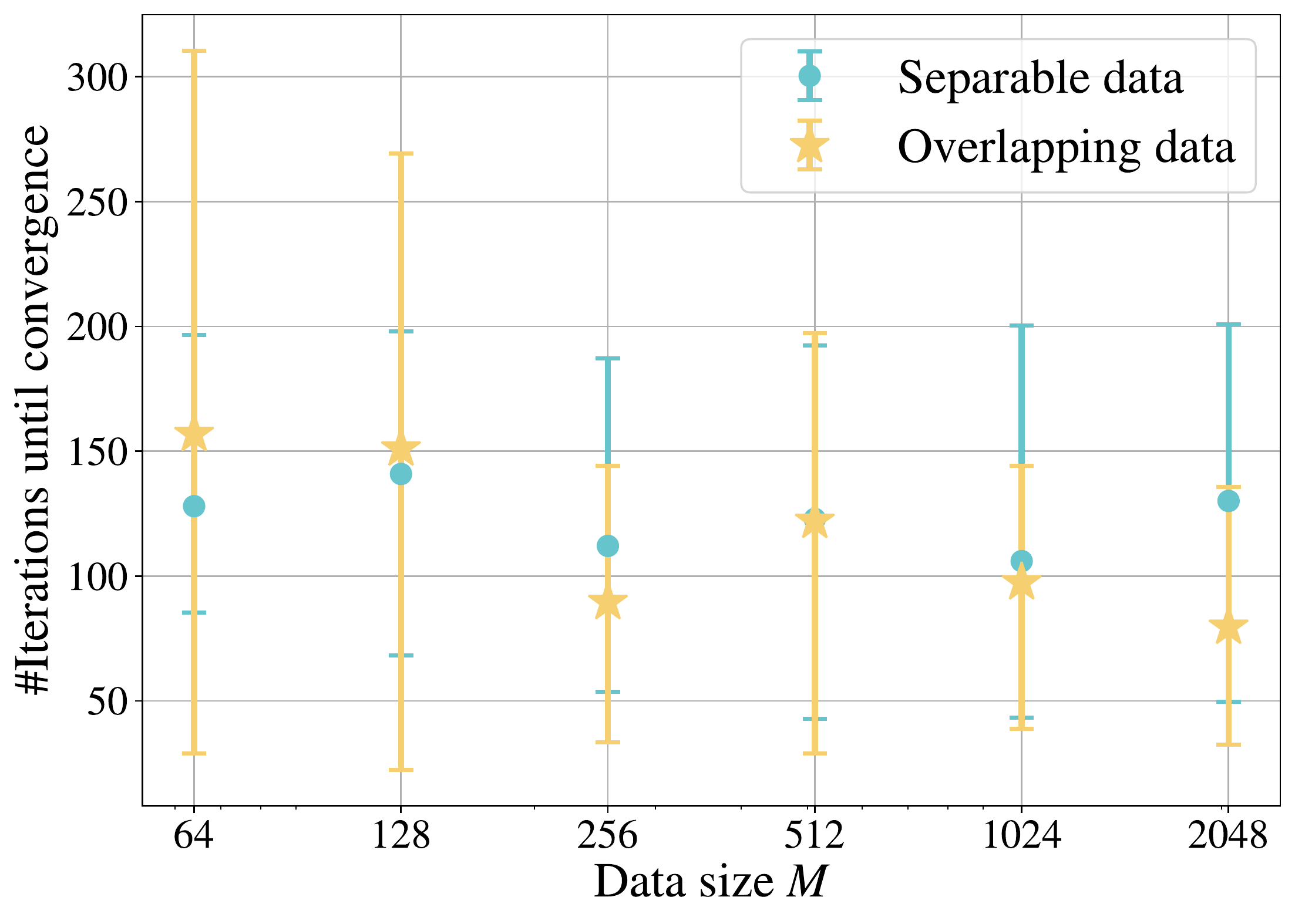}
  \caption{$M$-dependence}

\end{subfigure}%
\begin{subfigure}{.5\textwidth}
  \centering
  \includegraphics[width=\linewidth]{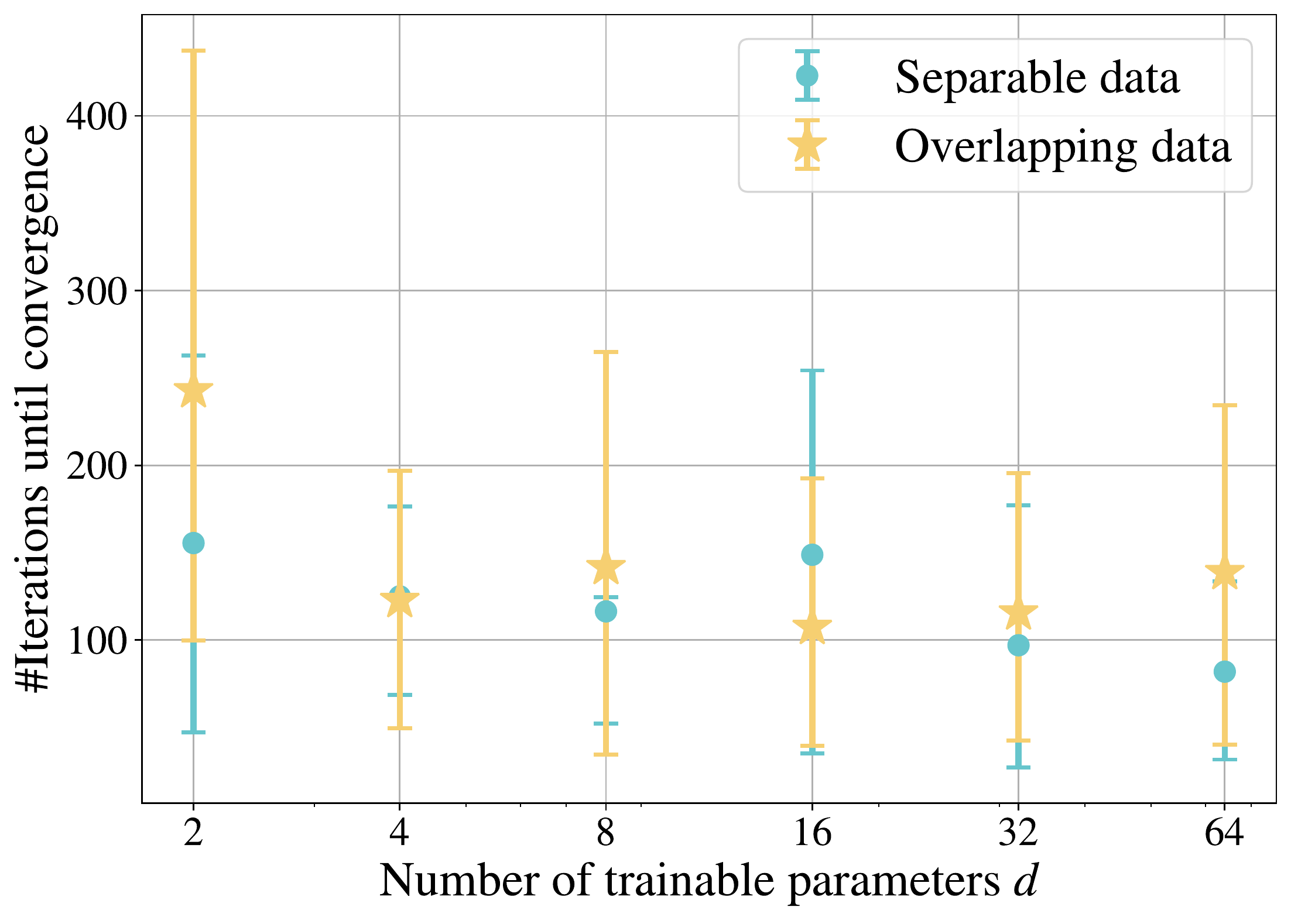}
  \caption{$d$-dependence}

\end{subfigure}
\caption{\textbf{$\boldsymbol{M}$- and $\boldsymbol{d}$-scaling of the approximate QSVM model:} The approximate QSVM is trained using a combination of the SPSA and SGD optimization algorithms and a statevector simulator. Convergence is defined as the iteration $T$ where the parameters fulfill $||\theta^T - \theta^{T-1}||/d < 10^{-4}$. The experiment is repeated for different (a) sizes $M$ of the 2-dimensional data set (with $d=8$ fixed) and (b) number of trainable parameters $d$ with ($M=256$ fixed). The total number of circuit evaluations for SPSA with a batch size of 5 is then given as $R_{tot} = 5\cdot 2 \cdot T$. Every experiment is repeated $n=10$ times, such that the markers correspond to the means and the error bars to the interval between the 15.9 and 84.1 percentile over these runs.}
\label{fig:approx_d_and_M}
\end{figure}
\begin{figure}[!htb]
\centering
\begin{subfigure}{.5\textwidth}
  \centering
  \includegraphics[width=\linewidth]{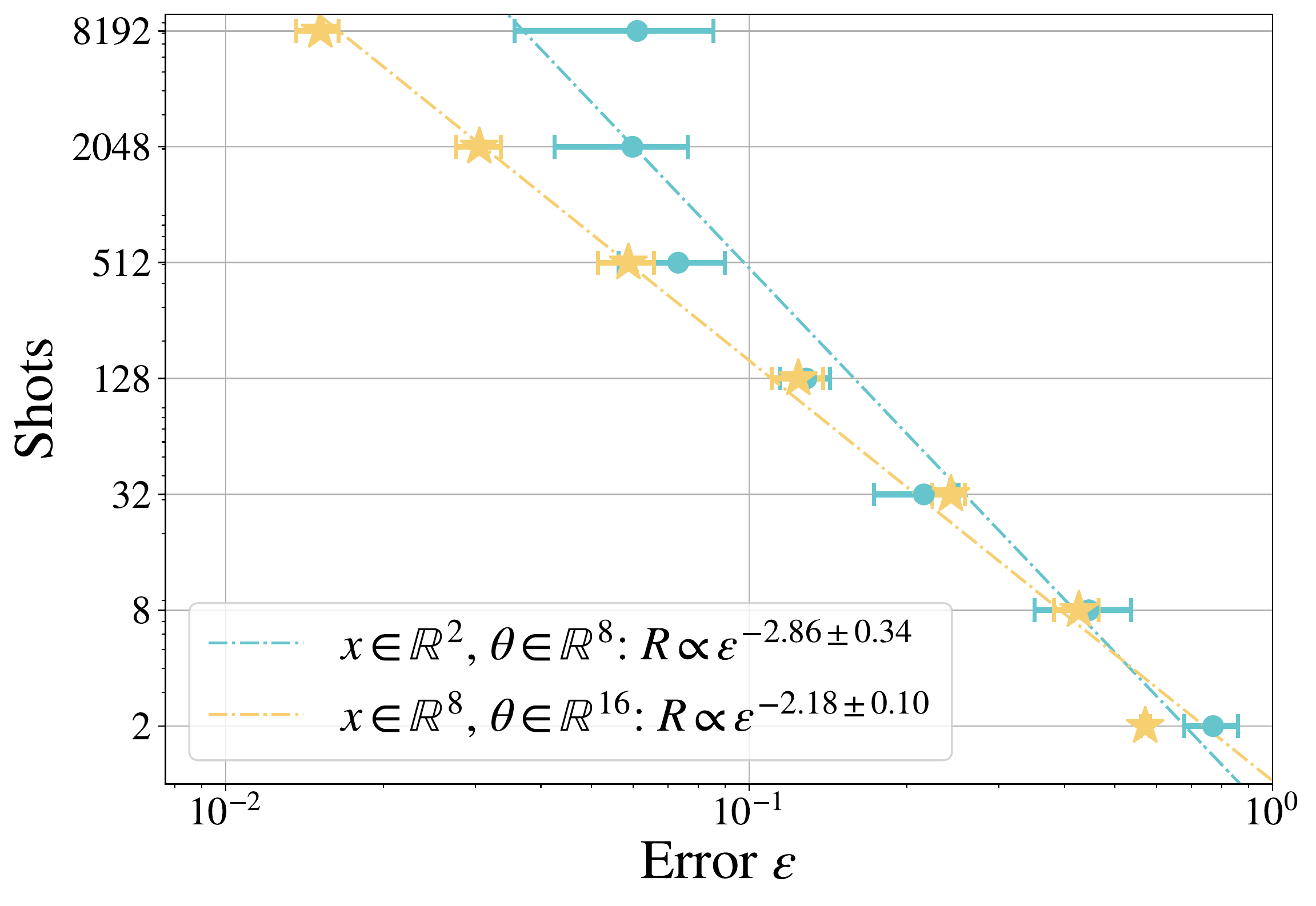}
  \caption{Linearly separable data}

\end{subfigure}%
\begin{subfigure}{.5\textwidth}
  \centering
  \includegraphics[width=\linewidth]{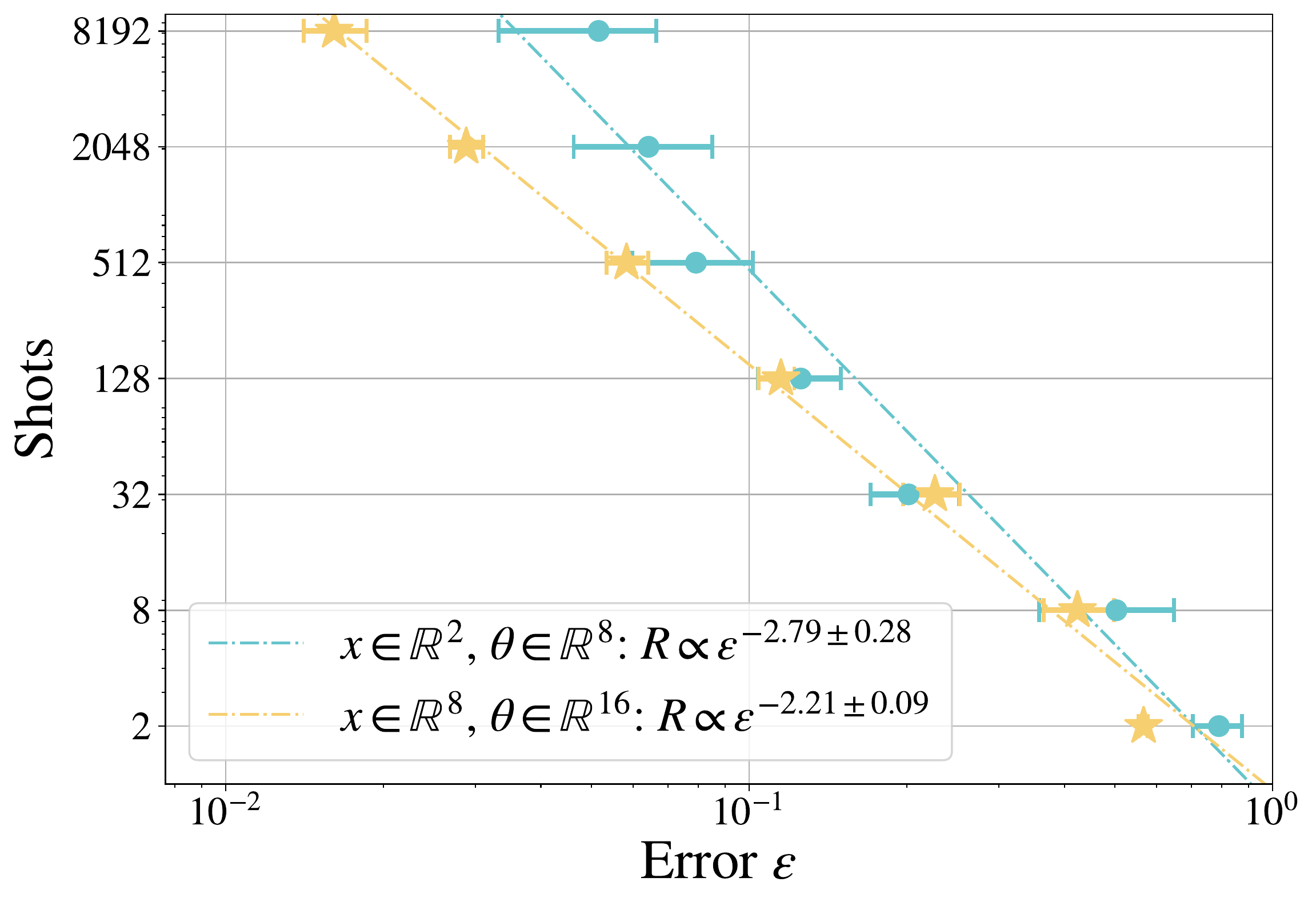}
  \caption{Overlapping data}

\end{subfigure}
\caption{\textbf{$\boldsymbol{\eps}$-scaling of the approximate QSVM model:} QASM-simulators with shots per expectation value ranging from 2 to 8192 are employed. The number of shots $R$ is plotted as a function of $\eps$ as defined in \Cref{eq:qnn_eps} for the training procedure detailed in the main text. A linear fit inside the log-log plot is used to determine the empirical exponent realized in our experimental setup. We analyze the case of 2-dimensional data with 8 trainable parameters (blue) and of 8-dimensional data with 16 trainable parameters (orange) both for separable and overlapping data. Every experiment has been repeated $n=10$ times and the markers shown are the means of the resulting errors, while the error bars mark the interval between the 15.9 and 84.1 percentile.}
\label{fig:qnn_exponent}
\end{figure}

To analyze the effect of finite sampling noise, the training is performed with Qiskit's shot based QASM-simulator~\cite{qiskit} using $R$ shots per expectation value. The optimization of the models is first performed for 1000 training steps, resulting in the noisy decision function $h_{\theta_{\! R}}(\mathbf{x})$. To enable a direct comparison to the ideal case, we minimize the loss using full gradient descent (i.e.~all parameters and all data in the training set are considered in the gradient evaluations) and a statevector simulator in an additional experiment. In this way, the noiseless decision function $h_{\theta_{\! \infty}}(\mathbf{x})$ is determined. Here, $\theta_{\! \infty}$ are the parameters that the statevector optimization ($R \to \infty$) converged to given the initial parameters $\theta_{\! R}$.
This procedure is repeated multiple times with different random initializations of the trained weights and varying $R$. 
For an error defined as
\begin{equation}
    \label{eq:qnn_eps}
    \eps \coloneqq \max_{\mathbf{x} \in X}\left|h_{\theta_{\! R}}(\mathbf{x}) - h_{\theta_{\! \infty}}(\mathbf{x})\right|,
\end{equation}
a least squares fit with respect to the number of shots is performed to conclude the experiment.
The result is presented in \Cref{fig:qnn_exponent}. From that experiment, the empirical complexity of the measurement shots needed for an $\eps$-accurate solution is found to be approximately given by
\begin{equation*}
    R_\text{tot} = \cO\left(\frac{1}{\eps^{2.9 \pm 0.3}}\right).
\end{equation*}


\section{Conclusion}
The results summarized in~\Cref{tab:qnn_qsvm_scaling} show that QSVMs can be trained in polynomial time with multiple training algorithms that differ in their computational complexity with respect to the size of the data set $M$ and the accuracy $\eps$ of the decision function. Both the dual- and primal-QSVM solvers benefit from the convexity of the optimization problem that guarantees convergence to the global optimum. For the dual optimization we find that the complexity scales as $\cO(M^{4.67}/\eps^2)$, which entails a significant overhead over training a classical SVM due to shot noise that is inherent in any quantum measurements. Conversely, the \Pegasos~algorithm provides a scaling that is independent of $M$ by solving the primal optimization problem using stochastic gradient descent. While this provides an advantage for training the classifier on large data sets, this algorithm scales worse with respect to the $\eps$. The cost of training the algorithm with quantum kernels increases to $\cO(1/\eps^{10})$ from $\cO(1/\eps^4)$ when using classical kernels. While this scaling looks daunting, it is important to remark that $\eps$ denotes the error between the noisy and noiseless decision function evaluated on the training set. However, when solving a classification task in practice, typically the generalization performance of the models on some test set is the quantity of interest. The relationship between $\eps$ and the generalization error is an open question and promising direction for future research.  

The best empirical scaling for training the QSVM optimization problem is found to be given by the approximate QSVM as the cost is also independent of $M$ and the dependence on $\eps$ is reduced to $\cO(1/\eps^3)$. However, by choosing the approximate QSVM, we lose any guarantees on finding the global optimum of the optimization problem as the loss landscape now becomes non-convex. In addition to ensuring that the quantum kernel does not suffer from exponential concentration, we now also need to make sure that the loss function and initial choice of parameters does not lead to trainability problems such as barren plateaus. Nevertheless, the efficient scaling which might render the approximate QSVM the only feasible variant in practical settings. 

Finally, we note that in this work we set the focus on the complexity of training QSVMs. We show that the optimization problem can be solved efficiently but the effect of statistical noise inherent to the quantum kernel introduces a polynomial overhead over the classical equivalent. It remains an open question whether quantum feature maps that provide a practical advantage over classical machine learning methods for real life classification tasks exist.

\paragraph{Author contributions} AT~and GG~contributed equally to this work. The theoretical analysis was mainly conducted by AT~while GG~performed most of the numerical experiments. The work was supervised by DS~and~SW. All authors contributed to the write-up of this manuscript.

\paragraph{Competing interests} The authors declare that there are no competing interests.

\paragraph{Acknowledgements}
We thank Amira Abbas for fruitful discussions about QSVMs,  Julien Gacon for help with the numerical experiments and Supanut Thanasilp for insightful discussions about exponential concentration of quantum kernels. GG acknowledges support by the NCCR MARVEL, a National Centre of Competence in Research, funded by the Swiss National Science Foundation (grant number 205602).
\paragraph{Data availability}
The code and data to reproduce all numerical experiments and plots in this paper is available at~\cite{GTSW_code}.

\bibliographystyle{arxiv_no_month}
\bibliography{notes}

\appendix

\section{Rigorous analysis of the dual approach}
\label{sec:appendix_complexity_dual}
In this appendix, we provide a rigorous mathematical treatment of the steps outlined in~\Cref{sec:complexity_dual}. In the first section, we derive how the error of the kernel matrix scales when its entries are subject to finite sampling noise. We then analyze how this error affects the result of the QSVM optimization problem in a second section.


\subsection{Justification of~\Cref{eq:latala_result}} \label{app:latala}

In a first step, define the kernel or Gram matrix $K \in \R^{M \times M}$ with entries
\begin{equation*}
	(K)_{ij} = k_{ij} = k(\mathbf{x}_i, \mathbf{x}_j) \quad \forall \, \mathbf{x}_i, \mathbf{x}_j \in X,
\end{equation*}
given by the kernel function evaluated for all combinations of data pairs in the training set. Writing $\ket{\psi(\mathbf{x})} = \cE(\mathbf{x}) \ket{0}$ for some unitary $\cE(\mathbf{x})$, we can approximate the expectation value $|\braket{\psi(\mathbf{x})}{\psi(\mathbf{x}')}|^2$ by preparing and measuring the state $\cE(\mathbf{x}')^\dagger \cE(\mathbf{x}) \ket{0}$ in the computational basis a total of $R$ times.

A more formal treatment of this method introduces the Bernoulli distributed random variable $\hat{k}_{ij}$ taking on the values
\begin{equation*}
	\hat{k}_{ij} = 
	\begin{cases}
		1 \quad \text{if the zero state $\ket{0}$ is recovered in the measurement}
		\\0 \quad \text{otherwise}
	\end{cases}.
\end{equation*}
The kernel is then equal to the true mean
\begin{align}\label{eq:kernel_expectation}
	k(\mathbf{x}_i, \mathbf{x}_j) = \E \left[\hat{k}_{ij} \right],
\end{align}
because the expectation value of the Bernoulli distribution is equal to the probability of $\hat{k}_{ij} = 1$. Fundamentally, we can only approximate this expectation value by the sample mean
\begin{align*}
	k_R(\mathbf{x}_i, \mathbf{x}_j) = \frac{1}{R} \sum_{l=1}^R \hat{k}_{ij}^{(l)},
\end{align*}
where the $\hat{k}_{ij}^{(l)}$ are the i.i.d.~outcomes of $R$ measurement shots performed for every single entry of $K$. Denote by $K_R$ the matrix with entries $k_R(\mathbf{x}_i, \mathbf{x}_j)$. Closely following \cite[Section V. C.]{Hubregtsen2021}, the goal is now to bound the operator distance between the ideal $K$ and obtainable $K_R$. 

To do so, define the error
\begin{align}\label{eq:noisy_kernel_matrix_distance}
	E_R \coloneqq K_R - K,
\end{align}
and observe that
\begin{align}
	\E \left[(E_R)_{ij}\right]
	= \E \left[ k_R(\mathbf{x}_i, \mathbf{x}_j) - k(\mathbf{x}_i, \mathbf{x}_j) \right]
	= \frac{1}{R} \sum_{l=1}^R \E \left[\hat{k}_{ij}^{(l)}\right] - \E \left[k(\mathbf{x}_i, \mathbf{x}_j)\right] = 0 \,,
	\label{eq:E_R_first_moment}
\end{align}
where the linearity of the expectation value and~\Cref{eq:kernel_expectation} have been used. The second moment of $E_R$ can be calculated as
\begin{align}
	\E \left[|(E_R)_{ij}|^2\right] 
	&= \E \left[ \left(\frac{1}{R} \sum_{\ell = 1}^R \hat{k}_{ij}^{(\ell)} - k_{ij}\right)^2\right] \nonumber\\
	&= \frac{1}{R^2} \, \E \left[ \left(\sum_{l=1}^R \hat{k}_{ij}^{(l)}\right)^2 \right] \underbrace{- \frac{2 k_{ij}}{R} \sum_{l=1}^R \underbrace{\E \left[\hat{k}_{ij}^{(l)} \right]}_{= k_{ij}} + k_{ij}^2}_{= - k_{ij}^2} \nonumber\\ 
	&= \frac{1}{R^2} \left(R \, \smash{\underbrace{\E \left[\left(\hat{k}_{ij}^{(l)}\right)^2\right]}_{(\star)}} + 2 {R\choose2} \underbrace{\E \left[ \hat{k}_{ij}^{(l)} \hat{k}_{ij}^{(m)}\right]}_{(\dagger)}\right) - k_{ij}^2,
	\label{eq:E_R_second_moment_star_and_dagger}
\end{align}
where in the step from the second to the third line, the sum inside the square is expanded into two terms collecting the two cases when both samples are identical and when they are different with $l \neq m$. The expression can be further simplified knowing
\begin{align*}
	(\star) = \Var \left[\hat{k}_{ij}\right] - \left( \E \left[\hat{k}_{ij}\right]\right)^2 = k_{ij} (1 - k_{ij}) + k_{ij}^2 = k_{ij},
\end{align*}
because $\hat{k}_{ij}$ is Bernoulli distributed and
\begin{align*}
	(\dagger) = \E \left[\hat{k}_{ij}^{(l)}\right] \E \left[\hat{k}_{ij}^{(m)}\right] = k_{ij}^2,
\end{align*}
because the $\hat{k}_{ij}^{(\cdot)}$ are independently and identically distributed. Plugging in these results back into \Cref{eq:E_R_second_moment_star_and_dagger} and making use of 
\begin{align*}
	{R \choose 2} = \frac{R!}{2! (R - 2)!} = \frac{R (R -1)}{2} \leq \frac{R^2}{2}
\end{align*}
finally yields
\begin{align}\label{eq:E_R_second_moment}
	\E \left[|(E_R)_{ij}|^2\right] &= \frac{1}{R^2} \left[ R \, k_{ij}+ 2 {R \choose 2} k_{ij}^2\right] - k_{ij}^2 \leq \frac{k_{ij}}{R} = \cO\left(\frac{1}{R}\right).
\end{align}
Analogously, it can be shown that the fourth moment scales as
\begin{align}\label{eq:E_R_fourth_moment}
	\E \left[|(E_R)_{ij}|^4\right] = \cO\left(\frac{1}{R^2}\right).
\end{align}
Knowing this, the following result from random matrix theory can be harnessed:
\begin{theorem}[{Latala's theorem~\cite[Theorem 2]{Lataa2005} and~\cite[Theorem 5.37]{Vershynin2009}}]
\label{th:latalas}
	For any matrix $X\in \R^{n\times n}$ whose entries are independent random variables $x_{ij}$ of zero mean, there exists a constant $c > 0$ such that
	\begin{equation}\label{eq:latala_bound}
		\E \left[ \norm{X}_2 \right] \leq c \left[ \max_{i \in [n]} \left(\sum_{j=1}^n \E \left[x_{ij}^2\right]\right)^{\frac{1}{2}} + \max_{j \in [n]} \left(\sum_{i=1}^n \E \left[x_{ij}^2\right]\right)^\frac{1}{2} + \left(\sum_{ij=1}^n \E \left[x_{ij}^4\right]\right)^\frac{1}{4} \right],
	\end{equation}
	where $\norm{\cdot}_2 = \sigma_{\max}(\cdot)$ denotes the operator norm induced by the Euclidean vector norm.
\end{theorem}
The matrix $E_R$ satisfies all of the assumptions in~\Cref{th:latalas} and by~\Cref{eq:E_R_first_moment,eq:E_R_second_moment,eq:E_R_fourth_moment}, yields
\begin{equation}\label{eq:latala_result_appendix}
	\E \left[\norm{E_R}_2\right] = \E \left[ \norm{K_R - K}_2 \right] = \cO\left( \frac{\sqrt{M}}{\sqrt{R}} \right).
\end{equation}
With this, a bound on the accuracy of the kernel matrix is established.

\subsection{Justification of~\Cref{eq:dual_result}} \label{app:daniel}
The next question is how the solution to the SVM optimization problem is affected by perturbations to the ideal $K$. The following treatment is based on \cite[Appendix C]{Liu2021}.
We show that we need $\cO(M^{8/3}/\varepsilon^2)$ measurement shots per kernel entry to ensure that $\left|h(\hat{\mathbf{x}}) - h_R(\hat{\mathbf{x}})\right| \leq \varepsilon$.
Given that $K$ is symmetric a total of $M(M+1)/2=\cO(M^2)$ unique entries need to be calculated, implying~\Cref{eq:dual_result}.

As a reminder, the goal is to solve the dual optimization problem~\Cref{opt:svm_soft_margin_feature_map_dual}, which can be expressed in terms of matrices. Define $Q$ as the matrix with entries $(Q)_{ij} \coloneqq y_i y_j \, k(\mathbf{x}_i, \mathbf{x}_j)$. The matrix $Q$ is positive semidefinite if $K$ is, since 
\begin{equation}\label{eq:q_matrix}
	Q = \diag(\mathbf{y}) \, K \, \diag(\mathbf{y}) \, ,
\end{equation}
for the vector of training labels $\mathbf{y} \in \R^M$, and $\diag(\mathbf{y})$ can be absorbed into the $\mathbf{v}$ in the definition of positive semidefiniteness of a matrix  $\mathbf{v}^\top Q \, \mathbf{v} \geq 0 \quad \forall \mathbf{v} \in \R^M \setminus \{\mathbf{0}\}$.
The kernel matrix $K$ is always positive semidefinite \cite[Section 2.2]{Hofmann2008}, implying that $Q$ must be too. 

The dual problem in~\Cref{opt:svm_soft_margin_feature_map_dual} can be written as the quadratic program~\cite[Eq.~(D19)]{Liu2021}
\begin{align}\label{opt:svm_l2_matrix_dual}
    \left \lbrace \begin{array}{rl}
    \min\limits_{\boldsymbol{\alpha} \in \R^M} & \frac{1}{2} \boldsymbol{\alpha}^\top \left(Q + \lambda \Id \right) \boldsymbol{\alpha} - \mathbf{1}^\top \boldsymbol{\alpha} \\
    \textnormal{s.t.} & \boldsymbol{\alpha} \geq 0 \, ,
    \end{array} \right.
\end{align}
where $\boldsymbol{\alpha}$ denotes the vector with components $\alpha_i$, $\mathbf{1} \in \R^M$ the all one vector and $\Id \in \R^{M \times M}$ the identity matrix.
Because the identity commutes with all matrices, $Q$ and $\lambda \Id$ can be simultaneously diagonalized. Since $Q$ is positive semidefinite, its eigenvalues are all greater than or equal to zero. So when $Q$ and $\lambda \Id$ are expressed in a basis where both matrices are diagonal, adding the two yields a diagonal matrix with entries greater than or equal to $\lambda$, implying that 
\begin{equation}\label{eq:lambda_lower_bound}
	\mu \geq \lambda
\end{equation}
for the the smallest eigenvalue $\mu$ of the matrix $(Q + \lambda\Id)$.

With this, the following theorem can be invoked to prove the robustness of the dual QSVM optimization when the kernel function is only evaluated up to some finite precision.
\begin{theorem}[Daniel's Theorem {\normalfont \cite[Theorem 2.1]{Daniel1973} \cite[Lemma 16]{Liu2021}}]\label{th:daniels}
	Let $x_0$ be the solution to the quadratic program
\begin{align}\label{opt:quadratic_program_theorem}
    \left \lbrace \begin{array}{rl}
    \min\limits_{x} &\frac{1}{2} x^\top K x - c^\top x \\
    \textnormal{s.t.} & G x \leq g \\
    & D x = d \, ,
    \end{array} \right.
\end{align}	
	where $K$ is positive definite with smallest eigenvalue $\mu > 0$ and the dimensions of the vectors $c, g, d$ and matrices $G, D$ are such that all operations are well-defined. Let $K'$ be a symmetric matrix such that $\norm{K' - K}_2 \leq \varepsilon < \mu$.\footnote{Positive definiteness of $K'$ is then implied.} Let $x_0'$ be the solution to \Cref{opt:quadratic_program_theorem} with $K$ replaced by $K'$. Then
	\begin{equation}\label{eq:daniels_result}
		\norm{x_0' - x_0} \leq \frac{\varepsilon}{\mu - \varepsilon} \norm{x_0}.
	\end{equation}
\end{theorem}

From the ground up, classification with QSVMs can be divided into two subsequent steps, training and prediction. For training, the quantum kernel matrix $K$ defined in \Cref{eq:svm_kernel_entries} is initially evaluated on a quantum computer. Then, the QSVM is fit by running the dual program in \Cref{opt:svm_l2_matrix_dual} on a classical computer, yielding a solution vector $\boldsymbol{\alpha}^\star$. For prediction, a new datum $\hat{\mathbf{x}} \in \R^d$ is assigned a class membership $\hat{y} \in \{-1, +1\}$ via the classification function $c_\text{SVM}(\hat{\mathbf{x}})$ in \Cref{eq:class_qsvm_dual}. To this end, the quantum computer has to be employed again to determine the quantum kernel value $k(\hat{\mathbf{x}}, \mathbf{x}_i)$ for all $\mathbf{x}_i$ in the training set $X$. 

From these two steps, it is clear that there are two separate instances in the QSVM algorithm where quantum kernels are evaluated and the statistical uncertainty inherent to quantum expectation values unavoidably enters. Because even on a fault tolerant quantum computer, fundamentally, only an approximate kernel function $k_R(\hat{\mathbf{x}}, \mathbf{x}_i)$ as defined in \Cref{eq:kernel_sample_mean} is feasible. Consequently, for training, $K$ has to be replaced by its approximation $K_R$, in turn leading to a noisy $Q_R$ analogous to \Cref{eq:q_matrix}. The solution to the dual~\Cref{opt:svm_l2_matrix_dual} when $Q$ is replaced by $Q_R$ is then denoted as $\boldsymbol{\alpha}_R^\star$. In the same way, the kernel evaluations in the decision function are replaced by evaluations of $k_R(\hat{\mathbf{x}}, \mathbf{x}_i)$ in the prediction step. Putting all of this together, a statement on the robustness of the overall QSVM with respect to measurement uncertainty can be made on the level of
\begin{equation*}
	h(\hat{\mathbf{x}}) \coloneqq \sum_{i=1}^M \alpha_i^\star y_i \, k(\hat{\mathbf{x}}, \mathbf{x}_i) \qquad \textnormal{and} \qquad
	h_R(\hat{\mathbf{x}}) \coloneqq \sum_{i=1}^M \alpha_{R,i}^\star y_i \, k_R(\hat{\mathbf{x}} \, , \mathbf{x}_i),
\end{equation*}
the ideal and noisy version of the term inside the sign function in the classification function \Cref{eq:class_qsvm_dual}. The goal is to show the robustness of QSVMs with respect to finite sampling noise by bounding the difference between these terms with high probability.

The following lemma is adapted from \cite[Lemma 19]{Liu2021}.
\begin{lemma}\label{lemma:qsvm_optimization_under_finite_sampling}
	Let $\varepsilon > 0$ and suppose $R = \cO(M^{8/3}/\varepsilon^2)$ measurement shots are taken for each quantum kernel estimation circuit. Furthermore, assume the setting of noisy halfspace learning (see Assumption~\ref{ass_noisy_halfspace}). Then, with fixed probability $p > \frac{1}{2}$, where the probability is stemming from the choice of random training samples and uncertainty coming from the finite sampling statistics, for every $\mathbf{x} \in \R^d$ we have
	\begin{equation*}
		|h_R(\mathbf{x}) - h(\mathbf{x})| \leq \varepsilon \, .
	\end{equation*} 
\end{lemma}

\begin{proof}
	Start by considering the operator distance between the noisy matrix $Q_R$ resulting from quantum kernel evaluations with a finite number of measurement shots $R$ per entry and the ideal matrix $Q$. Making use of the connection between $Q$ and $K$ in \Cref{eq:q_matrix}, we find
	\begin{align*}
		\norm{Q_R - Q}_2 
		&= \sup_{\mathbf{v} \neq 0} \frac{\norm{(Q_R - Q) \mathbf{v}}}{\norm{\mathbf{v}}} \\
		&= \sup_{\mathbf{v} \neq 0} \frac{\norm{\diag(\mathbf{y}) \, (K_R - K) \, \diag(\mathbf{y}) \mathbf{v}}}{\norm{\mathbf{v}}} \\
		&= \sup_{\mathbf{v} \neq 0} \frac{\sqrt{\sum_{i=1}^{M} \left|\sum_{j=1}^{M} y_i y_j (k_R(\mathbf{x}_i, \mathbf{x}_j) - k(\mathbf{x}_i, \mathbf{x}_j)) v_j\right|^2}}{\norm{\mathbf{v}}} \\
		&\leq \sup_{\mathbf{v} \neq 0} \frac{\sqrt{\sum_{i=1}^{M} \left|\sum_{j=1}^{M} (k_R(\mathbf{x}_i, \mathbf{x}_j) - k(\mathbf{x}_i, \mathbf{x}_j)) v_j\right|^2}}{\norm{\mathbf{v}}} \\
		&= \norm{K_R - K}_2\,,
	\end{align*}
	where in the inequality step uses $|y_i| = 1$ and the triangle inequality. Then, \Cref{eq:latala_result} implies
	\begin{equation}\label{eq:q_distance_expectation}
		\E \left[\norm{Q_R - Q}_2\right] = \cO \left(\frac{\sqrt{M}}{\sqrt{R}}\right).
	\end{equation}
	
	For the setting of noisy halfspace learning as defined in \cite[Lemma 14, Definition 15]{Liu2021}, it can be shown \cite[Remark 2]{Liu2021} that 
	\begin{equation}\label{eq:alpha_0_expectation}
		\E \left[ \norm{\boldsymbol{\alpha}^\star} \right] = \cO \left(M^{1/3}\right),
	\end{equation}
	where $\boldsymbol{\alpha}^\star$ is the solution to \Cref{opt:svm_l2_matrix_dual}. 
	With Markov's inequality
	\begin{equation*}
		\Pr\big[X \leq a \, \E[X]\big] \geq 1 - \frac{1}{a} \overset{!}{=} p'\,,
	\end{equation*}
	the statements \Cref{eq:alpha_0_expectation,eq:q_distance_expectation} in expectation can be transformed into statements in probability. To this end, choose $a$ such that $p' > 1 - \frac{1}{a} > \frac{1}{\sqrt[3]{2}}$ is fulfilled, yielding
	\begin{equation}\label{eq:prob_bound_eta}
		\eta \coloneqq \norm{Q_R - Q}_2 = \cO \left(\frac{\sqrt{M}}{\sqrt{R}}\right) \qquad \textnormal{and} \qquad  \norm{\boldsymbol{\alpha}^\star} = \cO\left(M^{1/3}\right) \, ,
	\end{equation}
	which hold with a probability greater than or equal to $p'$ respectively.
	
	Now, \Cref{th:daniels} can be leveraged as $Q$ is positive-definite because there is a lower bound on its smallest eigenvalue $\mu$ in \Cref{eq:lambda_lower_bound}. Then, for $\delta_i \coloneqq \alpha^\star_{R,i} - \alpha_i^\star$ and with probability at least $p'^2$, \Cref{eq:daniels_result} yields
	\begin{align*}
			\norm{\delta}
			&= \norm{\boldsymbol{\alpha}_R^\star - \boldsymbol{\alpha}^\star}\\ 
			&\leq \frac{\eta}{\mu - \eta} \norm{\boldsymbol{\alpha}^\star}\\
			& = \left(\frac{\eta}{\mu} + \cO\left(\frac{\eta^2}{\mu^2}\right)\right) \norm{\boldsymbol{\alpha}^\star}\\
			&= \cO\left(\frac{\sqrt{M}}{\sqrt{R}}\right) \cO\left(M^{1/3}\right)\\ 
			&= \cO\left(\frac{M^{5/6}}{\sqrt{R}}\right)\, ,
	\end{align*}
	where $R$ has to be chosen such that $\eta < \mu$, which is always feasible due to the lower bound on $\mu$.

	In a next step, denote by $\nu_i \coloneqq k_R(\hat{\mathbf{x}}, \mathbf{x}_i) - k(\hat{\mathbf{x}}, \mathbf{x}_i)$ for $i = 1, 2, ..., M$ the difference between the obtainable and ideal kernel function evaluated for pairs of the datum to be classified $\hat{\mathbf{x}}$ and the elements in the training set $\mathbf{x}_i$. As $\nu_i$ is a special case of the components in $E_R$ in \Cref{eq:noisy_kernel_matrix_distance}, $\E[\nu_i] = 0$ is implied by \Cref{eq:E_R_first_moment} and $\E[\nu_i^2] = \cO\left(\frac{1}{R}\right)$ by \Cref{eq:E_R_second_moment}. By making use of Chebyshev's inequality
	\begin{equation*}
		\Pr\left(|X - \E[X]| \leq a \sqrt{\Var[X]}\right) \geq 1 - \frac{1}{a^2} \overset{!}{=} p'\,,
	\end{equation*}
	and $\Var[X] = \E[X^2] - \E[X]^2$, the above expectation values can be converted into the statement that
	\begin{equation}\label{eq:prob_bound_nu}
		\norm{\nu} = \cO\left(\frac{\sqrt{M}}{\sqrt{R}}\right)
	\end{equation}
	holds with probability at least $p' > 1/\sqrt[3]{2}$, where the $\sqrt{M}$ scaling comes from the Euclidean norm and the fact that the vector $\nu$ has $M$ components $\nu_i$.
	
	Returning to the ideal and noisy decision functions, the bounds \Cref{eq:prob_bound_eta,eq:prob_bound_nu} can be combined to give
	\begin{align*}
		\left| h_R(\hat{\mathbf{x}}) - h(\hat{\mathbf{x}}) \right|
		&= \left| \sum_{i=1}^M \alpha^\star_{R,i} y_i \, k_R(\hat{\mathbf{x}}, \mathbf{x}_i) - \sum_{i=1}^M \alpha_i^\star y_i \, k(\hat{\mathbf{x}}, \mathbf{x}_i) \right| \\
		&= \left| \sum_{i=1}^M (\alpha_i^\star + \delta_i) y_i (k(\hat{\mathbf{x}}, \mathbf{x}_i) + \nu_i) - \sum_{i=1}^M \alpha_i^\star y_i \, k(\hat{\mathbf{x}}, \mathbf{x}_i) \right| \\
		&= \left| \sum_{i=1}^M \alpha_i^\star \nu_i + \delta_i k(\hat{\mathbf{x}}, \mathbf{x}_i) + \delta_i\nu_i \right| \\
		&\leq \left|\sum_{i=1}^{M}  \alpha_i^\star \nu_i \right| + \left|\sum_{i=1}^{M}  \delta_i k(\hat{\mathbf{x}}, \mathbf{x}_i) \right| + \left|\sum_{i=1}^{M}  \delta_i\nu_i \right| \\
		&\leq \norm{\alpha^\star}  \norm{\nu} + \sqrt{M} \norm{\delta} + \norm{\delta} \norm{\nu}\\
		&= \cO \left(\frac{M^{4/3}}{\sqrt{R}}\right),
	\end{align*}
	which is true with probability at least $p \coloneqq (p')^3 > \frac{1}{2}$ for any $\hat{\mathbf{x}} \in \R^d$. The penultimate step uses Cauchy-Schwarz and property $k(\hat{\mathbf{x}}, \mathbf{x}_i) \in [0, 1]$.
Hence, the desired accuracy $| h_R(\hat{\mathbf{x}}) - h(\hat{\mathbf{x}})| \leq \varepsilon$ for the QSVM decision function is obtained with probability greater than or equal to $p$ for
	\begin{equation}\label{eq:r_my_result}
		R = \cO\left(\frac{M^{8/3}}{\varepsilon^2}\right)
	\end{equation}
	measurement shots per quantum kernel evaluation.
\end{proof}
As remarked at the beginning of this appendix, this result justifies~\Cref{eq:dual_result} in the main text.

\begin{remark}
We note that this analysis is an improvement of the results derived in~\cite[Lemma 19]{Liu2021}, where they showed that $R_\text{tot} = \cO(M^6/\varepsilon^2)$ shots are necessary under the same assumptions.\footnote{ Note that the result~\cite[Section~V.C.]{Hubregtsen2021} claiming that $R_\text{tot} = \cO(M^3/\varepsilon^2)$ suffice is not comparable, since there, the analysis is restricted to the kernel matrix and $\varepsilon$ is defined differently as a result.}
\end{remark}

\section{Extended analysis of the primal approach}
\label{sec:appendix_eps_del}
In this appendix, we provide justifications for the steps discussed in~\Cref{sec:complexity_pegasos}. We first provide experiments supporting Assumption~\ref{ass_Pegasos_noise}. Then, we present a rigorous mathematical derivation of the connection between $\eps$ and $\delta$ in~\Cref{eq:epsilon-delta} in the main text.

\subsection{Empirical justification of Assumption~\ref{ass_Pegasos_noise}} \label{sec_justification_ass}
\label{sec:experiments_pegasos}
Assumption~\ref{ass_Pegasos_noise} claims that \Pegasos\, converges independently of the sampling noise when the number of measurement shots is large enough. To provide numerical evidence for this statement, we compare the evolution of the algorithm for different numbers of shots per kernel evaluation $R$. 

For the experiments presented in~\Cref{fig:pegasos_training_accuracy}, every optimization step of the \Pegasos~algorithm is performed with a noisy kernel value $k_R$ (see~\Cref{eq:kernel_sample_mean}) in line~\ref{pegasos:condition} of the algorithm. The accuracy plotted is then computed according to~\Cref{eq:class_qsvm_dual} evaluated for exact kernel evaluations, which facilitates a direct comparison of the different optimization runs.
We observe that for $R$ sufficiently large, the algorithm converges to a solution which is close to the one obtained by performing an infinite number of quantum circuit evaluations. Furthermore, visually the rate of convergence is not significantly slower for smaller but sufficiently large $R$. Taken together, these observations suggests that Assumption~\ref{ass_Pegasos_noise} is fulfilled.
\begin{figure}[!htb]
\centering
\begin{subfigure}{.5\textwidth}
  \centering
  \includegraphics[width=\linewidth]{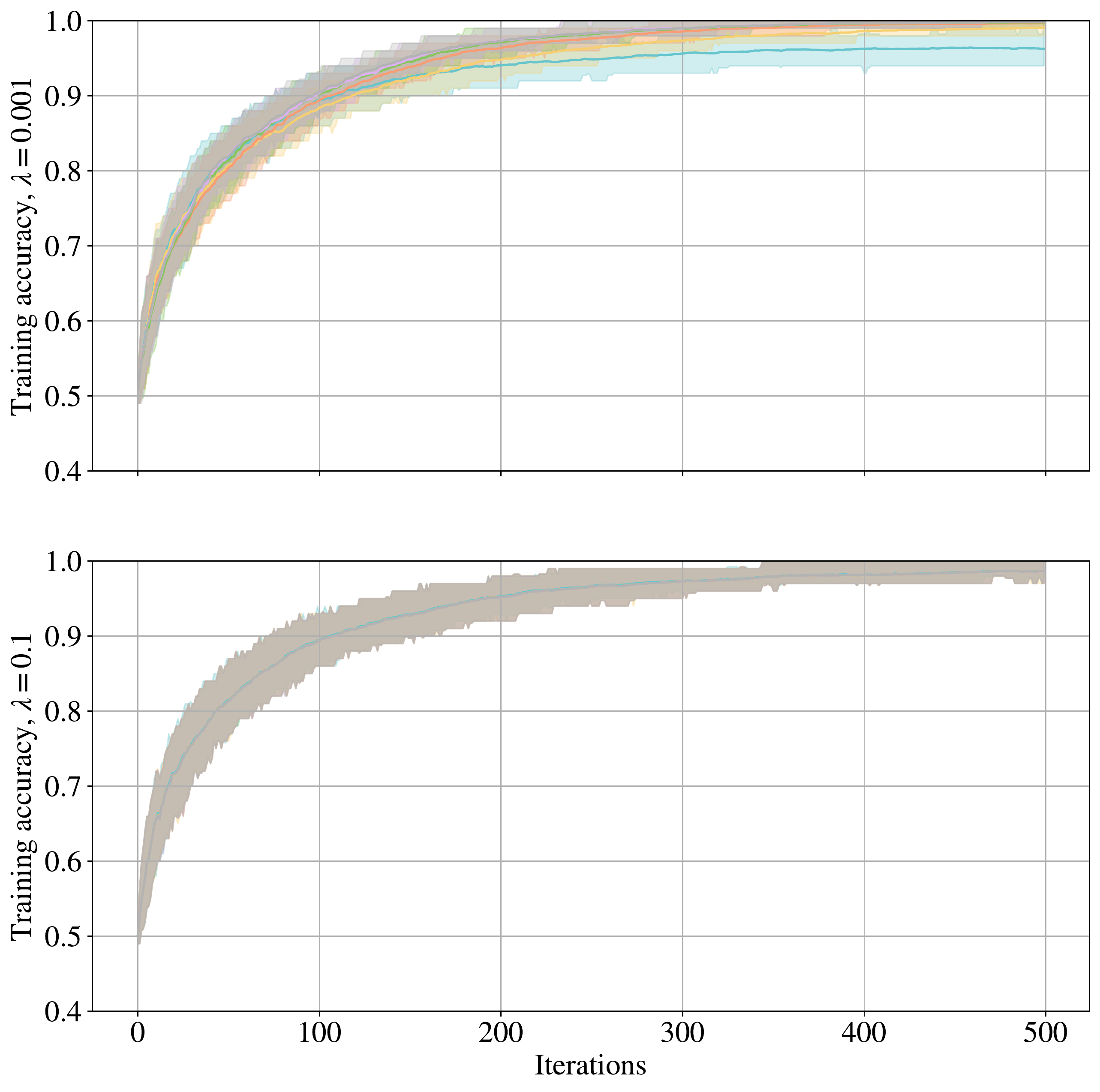}
  \caption{Linearly separable data}

\end{subfigure}%
\begin{subfigure}{.5\textwidth}
  \centering
  \includegraphics[width=\linewidth]{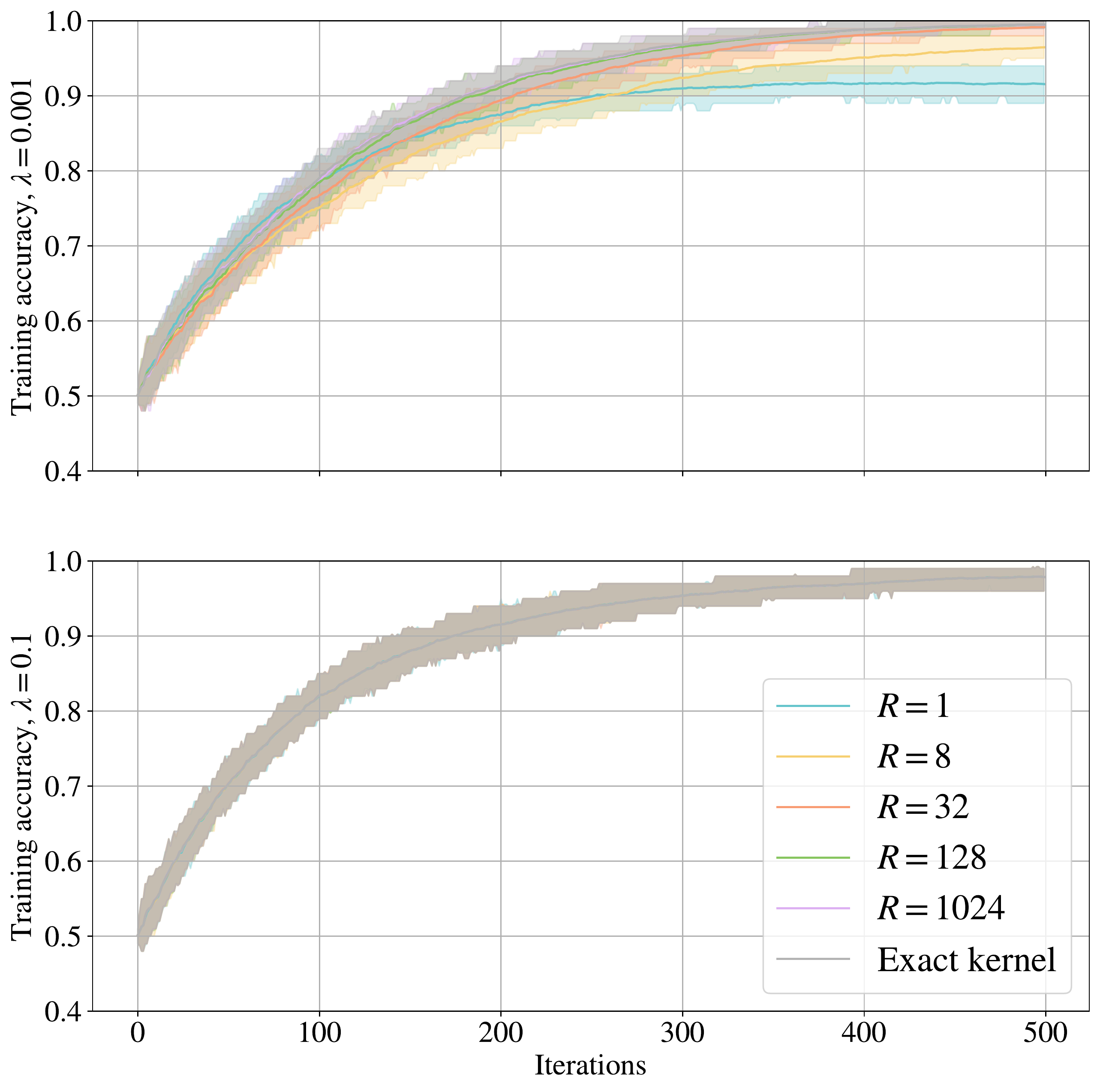}
  \caption{Overlapping data}

\end{subfigure}
\caption{\textbf{Training accuracy of \Pegasos:} The QSVM is trained with \Pegasos\ for different numbers of shots per kernel evaluation $R$. The class labels are then predicted for the whole training set using the exact kernels and compared to the ground truth. The resulting accuracy is plotted as a function of the iteration steps $t$. In the top plots, the regularization constant is $\lambda=1/1000$, while the lower plots are regularized with $\lambda=1/10$. Every experiment has been performed $n=100$ times with different random seeds, such that the lines shown are the means over the runs and the shaded areas mark the interval between the 15.9 and 84.1 percentile. }
\label{fig:pegasos_training_accuracy}
\end{figure}

In~\Cref{fig:pegasos_a_error}, we additionally present experiments probing how the coefficients behave under noisy kernel evaluations. The integer coefficients from \Pegasos\, are scaled by a factor of $\lambda/t$, similar to the sum in line~\ref{pegasos:condition}. Unsurprisingly, the error is quite large in the first few iterations as $t$ in the denominator is small. However, after around $t=40$ iterations the error stabilizes and in some cases even decreases. This provides further evidence that the convergence of \Pegasos\, is not strongly affected by finite sampling noise. 
\begin{figure}[!htb]
\centering
\begin{subfigure}{.5\textwidth}
  \centering
  \includegraphics[width=\linewidth]{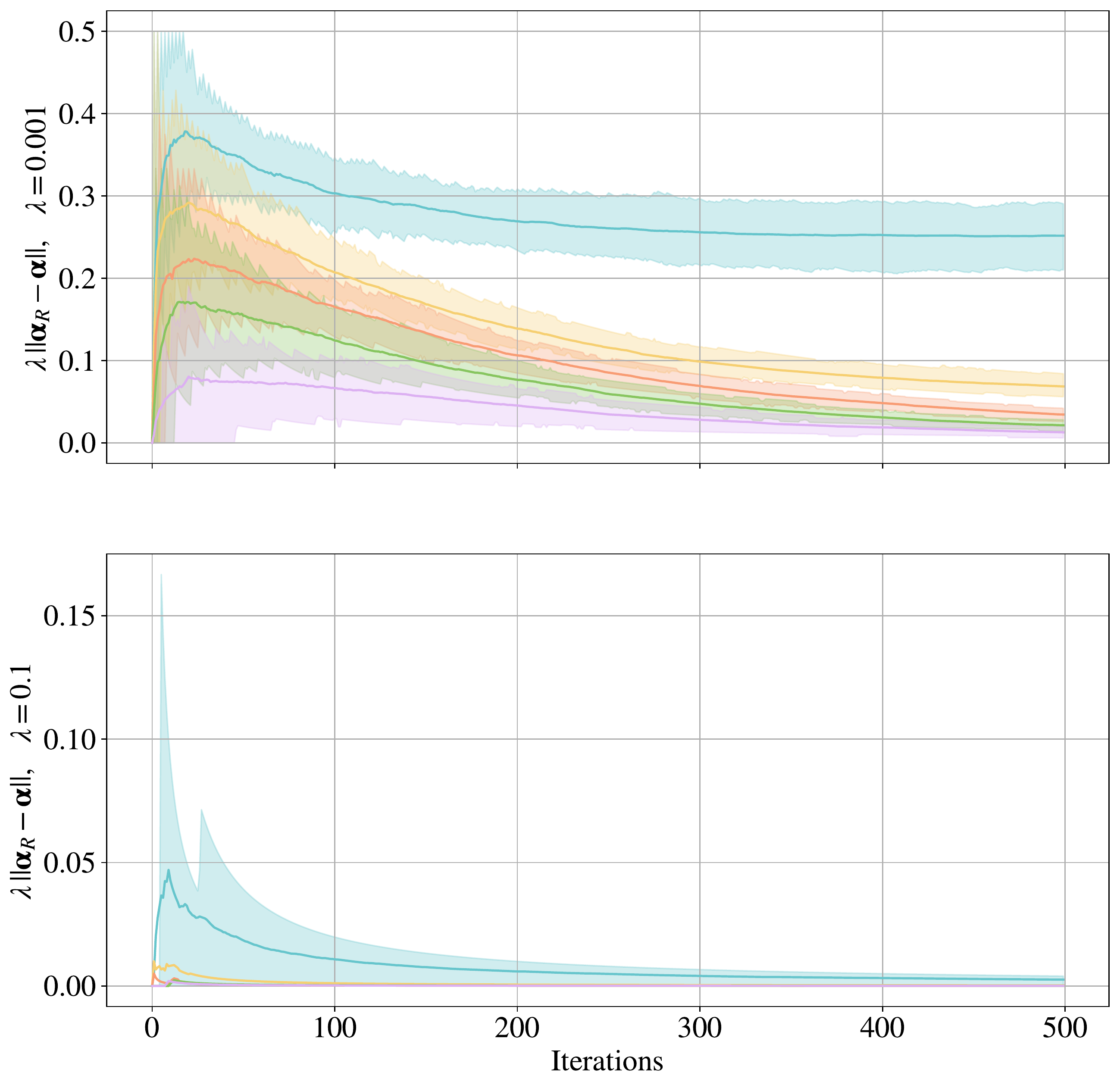}
  \caption{Linearly separable data}

\end{subfigure}%
\begin{subfigure}{.5\textwidth}
  \centering
  \includegraphics[width=\linewidth]{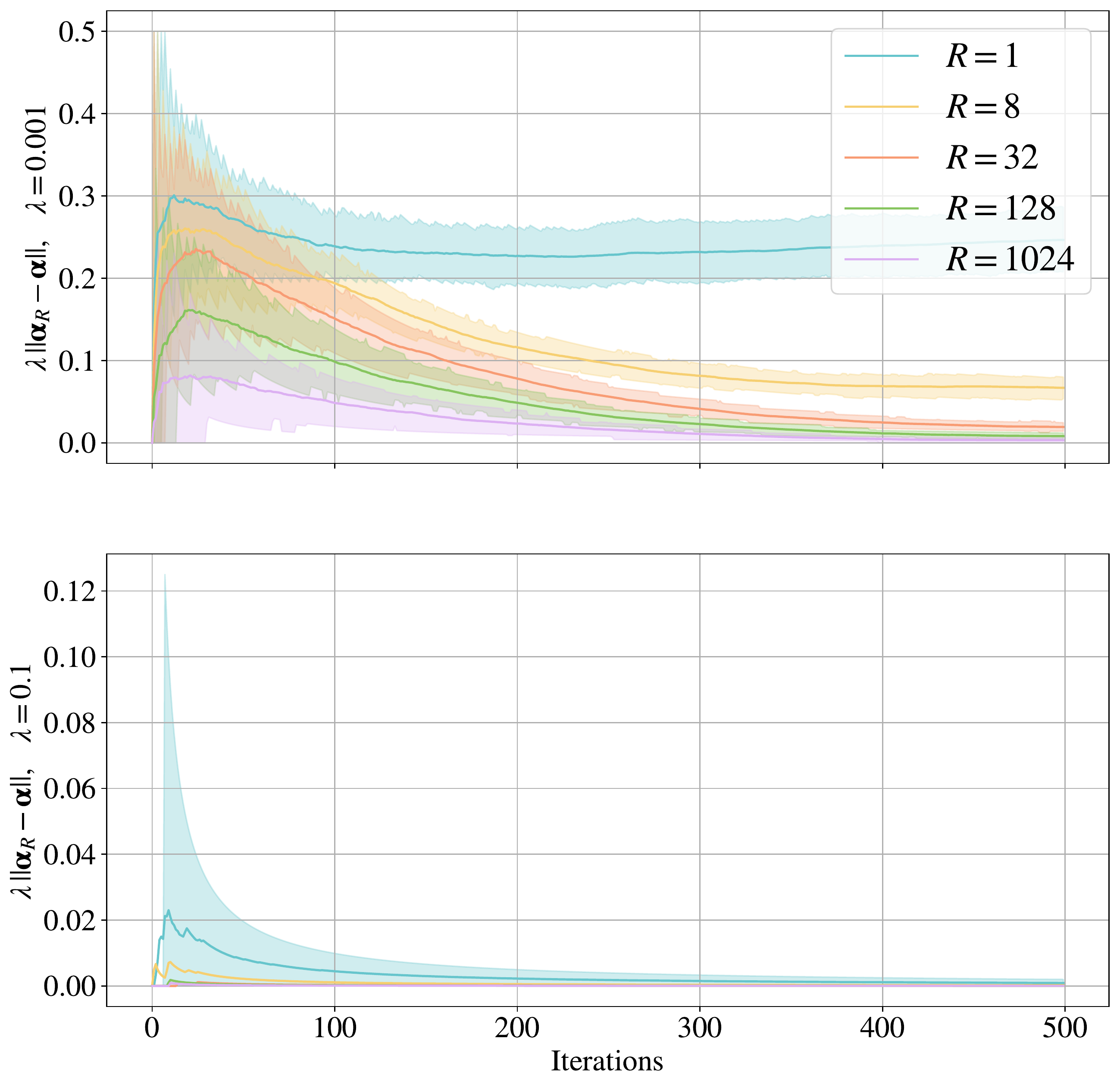}
  \caption{Overlapping data}

\end{subfigure}
\caption{\textbf{Error evolution of \Pegasos:} We plot the error in the coefficients $\boldsymbol{\alpha}$ for different numbers of shots per kernel evaluation $R$ as functions of the iteration steps $t$. The top plots were generated with a regularization constant $\lambda=1/1000$ while the lower plots are regularized with $\lambda=1/10$. Every experiment has been performed $n=100$ times with different random seeds such that the lines shown are the means over these runs and the shaded areas mark the interval between the 15.9 and 84.1 percentile. }
\label{fig:pegasos_a_error}
\end{figure}

\subsection{Justification of~\Cref{eq:epsilon-delta}}
\label{sec:appendix_pegasos_delta_eps}
In the following we denote by $f(\mathbf{w})$ the objective function of the primal optimization problem~\cref{opt:svm_soft_margin_primal_hinge_loss} where $\mathbf{w^\star}$ is the optimizer.
\begin{lemma}\label{lemma:strongly_convex}
The objective function $f(\mathbf{w})$ is $\lambda$-strongly convex in $\mathbf{w}$.
\end{lemma}
\begin{proof}
The hinge loss $\cL_\text{hinge}(y_i, \mathbf{w}^\top \mathbf{x} + b) = \max \left[0, 1 - y_i (\mathbf{w}^\top \mathbf{x} + b)\right]$ is convex and thus the sum 
\begin{equation*}
    g(\mathbf{w}) = \sum_{i=1}^M \cL_\text{hinge}(y_i, \mathbf{w}^\top \mathbf{x}_i + b)
\end{equation*}
is convex as well. Now $f$ is $\lambda$-strongly convex if and only if $f(\mathbf{w}) -  \frac{\lambda}{2}\norm{\mathbf{w}}$ is convex. But $f(\mathbf{w}) -  \frac{\lambda}{2}\norm{\mathbf{w}} = g(\mathbf{w})$ is convex and, thus, $f$ is indeed strongly convex.
\end{proof}
\begin{corollary} \label{cor_stronglyConvex}
Let $\varepsilon > 0$ and $\mathbf{w} \in \R^s$ be such that $f(\mathbf{w}) \leq f(\mathbf{w^\star}) + \varepsilon$. Then, $\norm{\mathbf{w} - \mathbf{w^\star}} \leq \sqrt{\frac{2\varepsilon}{\lambda}}$.
\end{corollary}
\begin{proof}
By definition of strong convexity and using Lemma~\ref{lemma:strongly_convex} we have
\begin{equation*}
    f(\mathbf{w}) \geq f(\mathbf{w^\star}) + \nabla f(\mathbf{w^\star})^T(\mathbf{w} - \mathbf{w^\star}) + \frac{\lambda}{2}\norm{\mathbf{w} - \mathbf{w^\star}}^2 \, .
\end{equation*}
Noting that the gradient vanishes at the minimum we find
\begin{equation*}
    \norm{\mathbf{w} - \mathbf{w^\star}}^2 \leq \frac{2}{\lambda}\big(f(\mathbf{w}) - f(\mathbf{w^\star}) \big) \leq \frac{2\varepsilon}{\lambda} \, .
\end{equation*}
\end{proof}

We are finally ready to prove the formal statement of~\Cref{eq:epsilon-delta}.
\begin{lemma}
For $\varepsilon = |h^\star(x) - h^N(x)|$ the distance between in the decision function $h^N(x) = \sum_{j=1}^{M}\alpha^N_jy_jk(x,x_j)$ after $N$ iterations of the \Pegasos~algorithm and the optimal decision function $h^\star$, we find
\begin{equation*}
    \varepsilon \leq \sqrt{\frac{2\delta}{\lambda}}\,,
\end{equation*}
where $\delta = |f(\mathbf{w^\star}) - f(\mathbf{w}^N)|$ is the deviation from the optimum.
\end{lemma}
\begin{proof}
Using the Kernel trick we write in general $\mathbf{w} = \sum_{i=1}^{M}\alpha_i y_i \mathbf{\phi}(x_i)$, allowing the following calculation:
\begin{align*}
    |h^\star(x) - h^N(x)| &= |\sum_{j=1}^{M}\alpha^\star_jy_jk(x,x_j) - \sum_{j=1}^{M}\alpha^N_jy_jk(x,x_j)| \\
    &= |\sum_{j=1}^{M}\alpha^\star_jy_j\phi(x)^T\phi(x_i) - \sum_{j=1}^{M}\alpha^N_jy_j\phi(x)^T\phi(x_i)| \\
    &= |\phi(x)^T\Big(\sum_{j=1}^{M}\alpha^\star_jy_j\phi(x_i) - \sum_{j=1}^{M}\alpha^N_jy_j\phi(x_i)\Big)| \\
    &= |\phi(x)^T(\mathbf{w}^N - \mathbf{w^\star})| \\
    &\leq \norm{\phi}\norm{\mathbf{w}^N - \mathbf{w^\star}} \\
    &\leq \sqrt{\frac{2\delta}{\lambda}} \, ,
\end{align*}
where the first inequality uses Cauchy-Schwarz and in the final follows from $\norm{\phi(x)} \leq 1$ as well as Corollary~\ref{cor_stronglyConvex}.
\end{proof}

\section{Generating artificial data}
In this appendix we provide the algorithm used to generate the training data for the numerical experiments in section~\Cref{sec_complexity}. In addition to the data set size $M$ and the margin $\mu$ separating the two classes, the algorithm requires a decision function $h_\theta$ as an input. For the approximate QSVM model, this decision function is given by~\eqref{eq:qnn_c}. To generate data for the experiments in~\Cref{app:dual_scaling,app:pegasos_scaling,sec:experiments_pegasos}, the variatonal form $\cW(\theta)$ in~\eqref{eq_h_theta} is chosen as the identity and thus trivial. In that case, the different realizations of the training set correspond to differing samples from the uniform distribution in~\Cref{algo:generating_artificial_data}, while the underlying decision boundary remains fixed. In~\Cref{sec_approximate_QSVM}, the $\theta$ in $\cW(\theta)$ are additionally randomly sampled.
\begin{algorithm}[H] 
	\caption{Generating artificial data}
	\label{algo:generating_artificial_data}
	\begin{algorithmic}[1]						
		\STATE \textbf{Inputs:} 
		\STATE $h_\theta(\mathbf{x})$ with fixed $\theta$,
		\STATE size of the data set $M \in \N$, 
		\STATE size of the margin $\mu \in \R$ (negative $\mu$ results in overlapping data).	
		\STATE
		\STATE $i = 0$
		\WHILE{$|y = -1| \leq \frac{M}{2} \land |y = 1| \leq \frac{M}{2}$}
		\STATE Sample $\mathbf{x} \sim \cU_{[0, 1]} \times \cU_{[0, 1]}$ once
		\STATE $\tilde{y} = h_\theta(\mathbf{x})$
		
		\IF{$\tilde{y} \leq - \frac{\mu}{2}$ \textbf{and} $\tilde{y} < + \frac{\mu}{2}$}
		\STATE $\mathbf{x}_i \gets \mathbf{x}$
		\STATE $y_i \gets -1$
		\STATE $i \gets i + 1$
		\ENDIF
		
		\IF{$\tilde{y} \geq + \frac{\mu}{2}$ \textbf{and} $\tilde{y} > - \frac{\mu}{2}$}
		\STATE $\mathbf{x}_i \gets \mathbf{x}$
		\STATE $y_i \gets +1$
		\STATE $i \gets i + 1$
		\ENDIF
		\IF{$\tilde{y} \geq + \frac{\mu}{2}$ \textbf{and} $\tilde{y} \leq - \frac{\mu}{2}$}
		\STATE Sample $p \sim \cU_{[0, 1]}$
		\IF{$p > \frac{1}{2}$}
		\STATE $y_i \gets +1$
		\ELSE
		\STATE $y_i \gets -1$
		\ENDIF
		\STATE $i \gets i + 1$
		\ENDIF
		\ENDWHILE
		\STATE
		\STATE \textbf{Output:} Balanced training data $X = \{\mathbf{x}_1,\mathbf{x}_2, ..., \mathbf{x}_M\}$ with labels $y = \{y_1, y_2, ..., y_M\}$.
	\end{algorithmic}
\end{algorithm}


\end{document}